\documentclass[12pt]{article}

\newcommand{\LaTeXtwoe}{\LaTeX\kern.15em 2${}_{\textstyle\varepsilon}$}
\newcommand{\AmS}{${\protect\the\textfont2 A}\kern-.1667em\lower
         .5ex\hbox{\protect\the\textfont2 M}\kern
         -.125em{\protect\the\textfont2 S}$}

\usepackage[
backend=biber,
style=alphabetic,
sorting=ynt
]{biblatex}

\usepackage{amsfonts}
\usepackage{amsthm}
\usepackage{xcolor}
\usepackage{physics}
\usepackage{tikz}
\usepackage[nooldvoltagedirection]{circuitikz}
\usetikzlibrary{quantikz}
\usepackage[margin=1in]{geometry}
\usepackage{amsmath}
\usepackage{amssymb}
\usepackage{mathtools}
\usepackage{complexity}
\usepackage{MnSymbol}
\usepackage[ruled,vlined,linesnumbered]{algorithm2e}
\usepackage{tabularx}
\usepackage[pdfpagelabels,bookmarks,hyperindex,hyperfigures]{hyperref}
\usepackage{cleveref}
\usepackage[toc]{appendix}

\usetikzlibrary{intersections}
\usetikzlibrary{circuits.logic.IEC}
\usetikzlibrary{shapes.gates.logic.IEC}
\usetikzlibrary{circuits.logic.US}
\tikzstyle{twoAnd}=[draw,and gate US,logic gate inputs=nn]
\tikzstyle{dot}=[fill,shape = circle,minimum size = 4pt,inner sep    = 0pt,text height  = 0pt,text depth   = 0pt]

\makeatletter
\newtheorem*{rep@theorem}{\rep@title}
\newcommand{\newreptheorem}[2]{%
	\newenvironment{rep#1}[1]{%
		\def\rep@title{#2 \ref{##1}}%
		\begin{rep@theorem}}%
		{\end{rep@theorem}}}
\makeatother
\newreptheorem{theorem}{Theorem}
\newreptheorem{lemma}{Lemma}
\newreptheorem{proposition}{Proposition}

\theoremstyle{plain}
\newtheorem{theorem}{Theorem}

\newtheorem{lemma}[theorem]{Lemma}
\newtheorem{proposition}[theorem]{Proposition}
\newtheorem{cor}[theorem]{Corollary}

\newtheorem{problem}[theorem]{Problem} 
 
\newtheorem{fact}[theorem]{Fact}
\newtheorem{observation}[theorem]{Observation}
\newtheorem{definition}[theorem]{Definition}

\newcommand{\CNOT}{\operatorname{CNOT}}
\newcommand{\CZ}{\operatorname{CZ}}

\newcommand{\parityL}{\mathsf{\oplus L}}
\newcommand{\BPAC}{\mathsf{BPAC}}
\newcommand{\LDAGParity}{\mathsf{LDAGParity}} 
 
\newcommand{\CNOTMult}{\mathsf{CNOTMult^*}}

\newcommand{\qBPAC}{\mathsf{qBPAC}}
\newcommand{\qAC}{\mathsf{qAC}}

\newcommand{\DAGParity}{\mathsf{DAGParity}}

\newcommand{\EAlg}{\mathfrak E}
\newcommand{\DAlg}{\mathfrak D}

\usepackage{enumitem}
\setlist[enumerate]{topsep=1ex,itemsep=0ex,partopsep=1.5ex,parsep=1ex}
\setlist[itemize]{topsep=1ex,itemsep=0ex,partopsep=1.5ex,parsep=1ex}
\setlist[description]{topsep=1ex,itemsep=0pt,partopsep=1.5ex,parsep=1ex}

\newcommand{\Dec}{\mathsf{Dec}}
\newcommand{\Rec}{\mathsf{Rec}}

\newcommand{\footremember}[2]{%
    \footnote{#2}
    \newcounter{#1}
    \setcounter{#1}{\value{footnote}}%
}
\newcommand{\footrecall}[1]{%
    \footnotemark[\value{#1}]%
}
\DeclareExpandableDocumentCommand{\hackcontrol}{O{}m}{|[phase,#1]| {}}

\title{Interactive quantum advantage with noisy, shallow Clifford circuits}
\author{Daniel Grier\footremember{iqc}{Institute for Quantum Computing, University of Waterloo, Canada}\footremember{cs}{\texttt{dgrier@uwaterloo.ca}.  Cheriton School of Computer Science, University of Waterloo, Canada}
\and Nathan Ju\footrecall{iqc} \footremember{uiuc}{\texttt{34ndju@gmail.com}. Department of Computer Science, University of Illinois, Urbana-Champaign, IL}
\and Luke Schaeffer\footrecall{iqc} \footremember{co}{ \texttt{lrschaeffer@gmail.com}. Department of Combinatorics and Optimization, University of Waterloo, Canada}
}

\date{}

\begin{document}
\maketitle

\begin{abstract}
\normalsize
	Recent work by Bravyi et al.\ constructs a relation problem that a noisy constant-depth quantum circuit ($\QNC^0$) can solve with near certainty (probability $1 - o(1)$), but that any bounded fan-in constant-depth classical circuit ($\NC^0$) fails with some constant probability. We show that this robustness to noise can be achieved in the other low-depth quantum/classical circuit separations in this area.  In particular, we show a general strategy for adding noise tolerance to the interactive protocols of Grier and Schaeffer. As a consequence, we obtain an unconditional separation between noisy $\QNC^0$ circuits and $\AC^0[p]$ circuits for all primes $p \geq 2$, and a conditional separation between noisy $\QNC^0$ circuits and log-space classical machines under a plausible complexity-theoretic conjecture.

	A key component of this reduction is showing average-case hardness for the classical simulation tasks---that is, showing that a classical simulation of the quantum interactive task is still powerful even if it is allowed to err with constant probability over a uniformly random input.  We show this is true even for quantum interactive tasks which are $\parityL$-hard to simulate.  To do this, we borrow techniques from randomized encodings used in cryptography.
\end{abstract}

\section{Introduction}

A major goal in quantum complexity theory is identifying problems which are efficiently solvable by quantum computers and not efficiently solvable by classical computers. If willing to believe certain conjectures, one can be convinced of this separation by the discovery of quantum algorithms that solve classically hard problems. For example, the belief that classical computers cannot efficiently factor integers contrasts with Shor's algorithm for factoring integers on a quantum computer \cite{shor}. However, a demonstration of Shor's algorithm on instances that are not efficiently solvable by classical computers would require quantum resources far out of reach of near-term capabilities. This has spurred developments in devising sampling problems that separate efficient, near-term quantum computers and classical computers like IQP circuit sampling \cite{iqp}, BosonSampling \cite{bosons}, and random circuit sampling \cite{rcs}. However, convincing evidence that {\it noisy} quantum computers outperform classical computers in these tasks suffers from the necessity of assuming some non-standard complexity-theoretic conjectures that are often native to each proposal.

Surprisingly, if you restrict to the setting of constant-depth circuits, a noisy, unconditional separation \emph{is} possible.  At first, these unconditional separations were known only in the noiseless setting.  That line of work was initiated by pioneering work of Bravyi, Gosset, and K{\"o}nig \cite{bgk} who showed a strict separation between constant-depth quantum circuits ($\mathsf{QNC^0}$) and constant-depth classical circuits with bounded fan-in gates ($\mathsf{NC^0}$).  The separation is based on the relation problem\footnote{Generally speaking, a relation problem is defined by a relation $R \subseteq \Sigma^* \times \Sigma^*$.  Given an input $x$, the task is to find some $y$ such that $(x,y) \in R$.} associated with measuring the outputs of a shallow Clifford circuit, which they called the Hidden Linear Function (HLF) problem due to certain algebraic properties of the output.  Furthermore, they show that $\NC^0$ cannot even solve this problem on average, a result which was later strengthened in several ways \cite{coudronstarkvidick:2018, legall:2019, bkst}.  Nevertheless, these works still assumed that the quantum circuit solving the task was noise free.

Follow-up work of Bravyi, Gosset, K{\"o}nig, and Tomamichel \cite{bgkt} showed that it was possible to encode the qubits of the quantum circuit in such a way that it preserved the separation while affected by noise. Interestingly, this was accomplished not by explicitly carrying out the quantum error correction procedure, but by simply measuring the syndrome qubits of the code and requiring the classical circuit to do the same.  In fact, their procedure provided a more-or-less general recipe for taking a constant-depth quantum/classical circuit separation and turning it into a separation in which the quantum circuit was also allowed noise.

This raises an obvious question:  how many circuit separations can we upgrade in this way?  $\mathsf{NC^0}$ circuits are fairly weak---they cannot even compute the logical AND of all input bits---and so we would like to show that even larger classes of classical devices cannot solve a problem that a noisy shallow quantum circuit can.

As a warm-up, we first consider the separation of Bene Watts, Kothari, Schaeffer, and Tal \cite{bkst}, which shows that constant-depth classical circuits with \emph{unbounded} fan-in gates ($\mathsf{AC^0}$) cannot solve the HLF problem on average.\footnote{The authors of that paper refer to their task as the ``Relaxed Parity Halving Problem,'' but it is still essentially the problem of measuring the outputs of a constant-depth Clifford circuit.}  Combining this result with the general error-correction recipe for relation problems, we arrive at the following result:

\begin{theorem}
There is a relation task solved by a noisy\footnote{We employ the same local stochastic noise model used in \cite{bgkt}.  We refer the reader to \Cref{sec:local_stochastic_noise} for those details.} $\QNC^0$ circuit with probability $1 - o(1)$ on all inputs. On the other hand, any probabilistic $\AC^0$ circuit can solve the problem on inputs of length $n$ with probability at most $\exp(-n^{\alpha})$ on a uniformly random input for some constant $\alpha > 0$.  
\end{theorem}

The result of Bene Watts et al.\ is the strongest known low-depth separation of its kind, but stronger separations are known for tasks which admit some amount of interactivity.  Consider the shallow Clifford measurement problem discussed above, where the measurements are made in two rounds.  In the first round, the quantum device is given the bases in which to measure some of the qubits and returns their measurement outcomes; and in the second round, the quantum device is given bases in which to measure the remaining qubits and returns their measurement outcomes.  Grier and Schaeffer \cite{gs} show that any classical device which can solve such problems must be relatively powerful.  More specifically, if the initial Clifford state is a constant-width grid state, then the classical device can be used to solve problems in $\NC^1$, and if the starting state is a poly-width grid state, then the classical device can be used to solve problems in $\parityL$.\footnote{See Section~\ref{sec:complexityclasses} for definitions of all the relevant complexity classes needed for the paper.} Because $\AC^0[p] \subsetneq \NC^1$ unconditionally, the above interactive task can be solved by a $\QNC^0$ circuit but not an $\AC^0[p]$ circuit, i.e., an $\AC^0$ circuit with unbounded $\MOD_p$ gates for some prime $p$.

One of the contributions of this work is massaging the noisy circuit separation recipe for relations problems into a recipe for \emph{interactive} problems as well.  Starting from an interactive protocol which exhibits a separation with a noise-free quantum circuit, there are three key steps to upgrade the separation to the noisy setting:
\begin{enumerate}
\item \emph{Augment the interactive protocol with the surface code encoding of Bravyi et al.\ \cite{bgkt}.}  This is straightforward, but it's worth noting that it changes the problem definition---not just because there are more physical qubits due to the encoding, but because we cannot prepare the initial state exactly or decode the syndrome in constant depth. As for the relational case, the burden of these steps is offloaded into the problem definition.
\item \label{step:hard} \emph{Show classical average-case hardness.}  That is, show that even when the classical circuit simulating the interactive protocol is allowed to err with some constant probability over a uniformly random input, it can still be leveraged to solve a hard problem (e.g., a problem in $\NC^1$ or $\parityL$).  This step is the most involved, and new ideas will be required to upgrade existing interactive separations in this way.\footnote{There is a sense in which classical average-case hardness is not strictly required to construct an interactive protocol demonstrating a quantum/classical separation robust to noise.  That is, one could require that any classical simulation of the protocol be correct with high probability on \emph{all} inputs.  In this case, simply combining the results of \cite{bgkt} with an amplified version of \cite{gs} suffice to show these types of simulations can be leveraged to solve hard problems.  However, we believe (much like \cite{bgkt}) that these results are more convincing when a weaker requirement is made on the classical simulation---namely, that the classical device need only succeed with high probability over some reasonable distribution of inputs.  In this setting, average-case hardness is required.}
\item \emph{Connect to separations of classical complexity classes.}  In some cases, this will lead to an unconditional separation between noisy shallow quantum circuits and shallow classical circuits, and in some cases this will lead to a conditional separation.  We note that these separations will not be identical to those obtained in Ref~\cite{gs} due to the fact that we use quasipolynomial-size circuits to decode the syndrome qubits of the surface code. 
\end{enumerate}

Fortunately, it was shown in Ref~\cite{gs} that Step~\ref{step:hard} holds\footnote{Although it is not strictly required, we prove a slightly stronger average-case hardness in \Cref{etolerance}.} for the $\NC^1$-hardness result.  We immediately obtain the following separation:
\begin{theorem}
\label{informalnoisyseparation}
There is a two-round interactive task solved by a noisy $\QNC^0$ circuit with probability $1 - o(1)$ on all inputs. Any probabilistic $\AC^0[p]$ circuit (for primes $p \geq 2$) fails the task with some constant probability on a uniformly random input. 
\end{theorem}

Unfortunately, Step~\ref{step:hard} is left as an open question in Ref~\cite{gs} for the $\parityL$-hardness result.  The second major contribution of this paper (and main technical result) is to show that we can, in fact, obtain average-case hardness for this setting:


\begin{theorem}
\label{informalparityl}
There is a two-round interactive task solved by a $\QNC^0$ circuit with certainty.
Let $\mathcal{R}$ be a classical oracle solving the task with probability $1 - \delta$. There exists $\delta > 0$ such that when $\mathcal{R}$ is implemented by a broad family\footnote{The actual broad family of devices we require are those that are ``rewindable" which we discuss in \Cref{sec:rewindoracle}.} of classical devices (e.g., log-space Turing machines, log-depth and bounded fan-in classical circuits) then 
$$
\parityL \subseteq (\BPAC^0)^{\mathcal{R}}.
$$
\end{theorem}

The proof of this theorem borrows an idea from cryptography called {\it randomized encodings}.  In particular, we will employ the construction of Applebaum, Ishai, and Kushilevitz \cite{aik} which randomizes instances of the following problem---given a directed acyclic graph (DAG), determine the parity of the number of paths from vertex $s$ to vertex $t$.  In fact, we will use that this problem reduces to the $\parityL$-hardness result in \cite{gs}.  Importantly, we show that when we compose the randomized encoding with the rest of the reduction, the distribution over inputs in the promise will be fairly uniform.  This leads to a general way to boost the randomization in worst-to-average-case reductions using the framework in \cite{gs}.

Combining the recipe for interactive circuit separations with \Cref{informalparityl}, we obtain the following consequence:
\begin{theorem}
\label{thm:parityL_noise_sep}
There is a two-round interactive task solved by a noisy $\QNC^0$ circuit with probability $1 - o(1)$ on all inputs.  Assuming $\parityL \not\subseteq (\qBPAC^0)^{\cL}$, any probabilistic log-space machine fails the task with some constant probability over a uniformly random input. 
\end{theorem}

Let us briefly unpack the $\parityL \not\subseteq (\qBPAC^0)^{\cL}$ assumption.  First, consider the plausible assumption that $\parityL \not\subseteq \cL$.  An $\cL$ machine is deterministic, while a $\parityL$ machine is non-deterministic and accepts if the parity of accepting paths is zero.  On the other hand, it is well-known that the parity function is not in $\qBPAC^0$ (i.e., randomized $\AC^0$ circuits of quasipolynomial size). Therefore, one might also expect that $(\qBPAC^0)^{\cL}$ is insufficiently powerful to compute $\parityL$ functions.\footnote{This assumption is slightly nonstandard; specifically, we would hope for a weaker assumption like $\parityL \not\subseteq (\BPAC^0)^{\cL}$. The barrier to proving our theorem under this assumption is that we require a {\it polynomial-size} low-depth classical circuit to decode the surface code. When the number of physical qubits per logical qubit is of order $O(\log n)$, there is a polynomial-size $\AC^0$ circuit for this task, but in our constructions, we have that the number of physical qubits per logical qubit is of order $\omega(\log n)$ to achieve a constant noise rate, and to the authors' knowledge, it is not known how to decode the surface code in $\AC^0$ (or even $\NC^1$) under these constraints.}

We do not attempt to give an exhaustive list of separations obtainable from combining \Cref{informalparityl} and the recipe for interactive circuit separations.  Much like the results of \cite{gs}, there is an inherent trade-off to the separation.  We can weaken the assumption at the expense of weakening the separation. For example, instead of a conditional separation between $\QNC^0$ and $\cL$ as in \Cref{informalnoisyseparation}, one could also make the analogous statement about $\QNC^0$ and $\NC^1$ under the assumption $\parityL \not\subseteq (\qBPAC^0)^{\NC^1}$.

\subsection{Open Problems}

Our work still leaves several unresolved questions.  We show average-case hardness results for classical devices solving quantum interactive tasks---if the classical machine only errs with small probability on a uniformly random input, it can be leveraged to solve hard problems.  We ask what happens in the error regime below this threshold. Can the success probabilities be exponentially reduced by parallel repetition of the same problems? Direct product theorems are often useful in proving these types of error amplifications, but it is unclear how they apply in this setting.  Parallel repetition could also \emph{improve} the success probability of the noisy quantum circuit if we only require that some fraction of the instances are solved correctly.  This could boost the error probability from inverse quasipolynomial to inverse exponential.

One could also approach this problem from the other direction by showing that there are circuits that can tolerate even higher amounts of error.  More generally, we ask what is the optimal amount of allowable error for each problem.  We show error thresholds of $29/30$ for classical devices solving $\NC^{1}$ problems, and $420/421$ for $\parityL$ problems.  Surely, these bounds are not tight.  How far can they be improved?

\section{Preliminaries}

This section discusses much of the background needed for this paper.  Readers familiar with previous results in this area, particularly \cite{bgkt} and \cite{gs}, can comfortably skip this section, except for \Cref{sec:problem_statement} where we formally define the main task we will consider for the rest of the paper.

\Cref{sec:complexityclasses} briefly touches on the low-depth circuit classes relevant to this paper.  \Cref{sec:local_stochastic_noise} explains the local stochastic noise model \cite{bgkt}, i.e., the type of error we allow in the quantum circuit. \Cref{sec:surface_code} describes the surface code (also following \cite{bgkt}) and its behavior under local stochastic noise.  \Cref{sec:problem_statement} defines a shallow Clifford circuit measurement task that unifies both the relation and interactive problem statements of \cite{bgk} and \cite{gs}.  Nevertheless, to make this distinction clear, we discuss the differences in relation and interactive tasks in \Cref{sec:rel_and_int}.  Finally, in \Cref{sec:mbqc} we discuss the connection between our problem and measurement-based quantum computation.

\subsection{Weak complexity classes}
\label{sec:complexityclasses}

For reference, we list the weak complexity classes of interest in this paper:

\begin{itemize}
	\item $\NC^i$: $\log^i$-depth circuits with bounded fan-in AND/OR/NOT gates.
	\item $\AC^i$: $\log^i$-depth circuits with unbounded fan-in AND/OR/NOT gates.
	\item $\AC^i[p]$: $\AC^i$ circuits with $\MOD_p$ gates. The $\MOD_p$ gate outputs 1 iff the sum of the inputs bits is 0 mod $p$.
	\item $\BP \mathcal{C}$: $\mathcal{C}$ circuits (e.g., $\NC^0$, $\AC^0$, etc.) that have access to random bits and two-sided bounded error.
	\item $\mathsf{q}\mathcal{C}$: $\mathcal{C}$ circuits of quasipolynomial (i.e., $\exp( \log^{O(1)} n)$) size.
	\item $\cL$: log-space Turing machines. 
	\item $\parityL$: Non-deterministic $\cL$ machines such that the accepting condition is that the number of accepting paths is odd. 
	\item $\QNC^i$: $\log^i$-depth \emph{quantum} circuits with arbitrary one and two-qubit gates.
\end{itemize}

In addition, we have the following inclusions that are either proven strict ($\subsetneq$) or believed to be strict ($\subseteq$): $\NC^0 \subsetneq \AC^0 \subsetneq \AC^0[p] \subsetneq \NC^1 \subseteq \cL \subseteq \parityL$.  We note that the inclusion $\AC^0[p] \subseteq \NC^1$ is only known to be strict when $p$ is prime.

\subsection{Local Stochastic Noise Model}
\label{sec:local_stochastic_noise}

While a noise-free quantum computation can reliably execute a sequence of operations, a noisy quantum computation may have sources of errors that corrupt several key parts of the computation including state initialization, gate execution, and measurement. To capture these sources of error, we consider the \emph{local stochastic quantum noise} model \cite{fgl, bgkt}. Under this model, random errors occur at each timestep of the execution of a quantum circuit. For example, a gate error occurs when random noise enters the computation prior to the execution of the gate. Similarly, an erroneous measurement outcome is modeled by random noise affecting the state of the system right before measurement. 

The types of random noise that we consider are random Pauli errors on each qubit. For a Pauli error $E \in \{I,X,Y,Z\}^{\otimes n}$, we borrow the convention of $\text{Supp}(E) \subseteq [n]$ to denote the subset of indexed qubits for which $E$ acts by a $X$, $Y$, or $Z$. In other words, $\text{Supp}(E)$ is the subset of qubits on which $E$ acts non-trivially. Local stochastic noise is parameterized by the noise rate $p$:

\begin{definition}
Let $p \in [0,1]$. A random $n$-qubit Pauli error $E$ is \emph{$p$-local stochastic} if
	\begin{equation}
		{\rm Pr}[F \subseteq {\rm Supp}(E)] \leq p^{|F|} \quad \forall \, F \subseteq [n]
	\end{equation} 
\end{definition}

Notice that this allows distant qubits to have correlated errors, but the probability that $k$ qubits are corrupted simultaneously decreases exponentially in $k$.  When we say that a layer of local stochastic noise $E$ is sampled with noise rate $p$, we use the notation $E \sim \mathcal{N}(p)$.

\subsection{The 2D surface code}
\label{sec:surface_code}

The 2D surface code is a CSS-type error correcting code that encodes one logical qubit into $m$ physical qubits on a 2D lattice. For a detailed discussion of the construction we employ throughout this paper, we refer the reader to Section IV of \cite{bgkt}. We will abstract away the physical surface code and henceforth denote the encoded version of a state with a line over the state vector, e.g., the logical $\ket{0}$ state is encoded as the $\ket*{\overline{0}}$ state. We follow the same convention when speaking of encoded circuits and measurement observables. If $\mathcal{Y}$ is the physical measurement outcome over multiple codeblocks, we denote $\mathcal{Y}=\mathcal{Y}^1...\mathcal{Y}^n$ with each $\mathcal{Y}^i$ the $m$ outcomes of the $i$'th codeblock. The space of binary physical measurement outcomes on one codeblock forms a linear subspace called the codespace, and we refer to it by $\mathcal{L}$.

A standard quantum computation begins with qubits prepared in a basis state, e.g., multiple copies of $\ket{0}$. However, the surface code must begin with an {\it encoded} basis state, $\ket*{\overline{0}}$, so we require a constant-depth procedure to produce such a state.  We can do this, albeit at the cost of extra ancilla qubits and a Pauli recovery operator dependent on the measurement outcomes of the ancillae. 

\begin{fact}
	\label{bellprep}
	(Basis state preparation, Theorem 23 in \cite{bgkt}) There is a constant-depth Clifford circuit on $2m + m_{anc}$ qubits that measures $m_{anc}$ qubits with measurement outcome $s \in \{0,1\}^{m_{anc}}$ and leaves the remaining $2m$ qubits in the state $\Rec(s)\ket*{\overline{0}}^{\otimes 2}$ for some Pauli operator $\Rec$ completely determined by $s$.\footnote{We note that the original construction of \cite{bgkt} actually creates an encoded Bell state $\ket*{\overline{\Phi}}$ instead of two zero states $\ket*{\overline{0}}^{\otimes 2}$ (up to a Pauli recovery operator), but these two states are equal up to a constant-depth Clifford, so these protocols are essentially identical using \Cref{constantdepthclifford} and the fact that Pauli operators can be ``pushed through" Cliffords, see \Cref{f}.}  Furthermore, the error on the final encoded basis state can be taken to be local stochastic noise with error parameter scaling polynomially in the error parameter for the circuit (but independent of $m$ and $m_{anc}$).
\end{fact}

Following basis-state preparation, we would like to perform a constant-depth Clifford circuit on the surface code.

\begin{fact}[Constant-depth Cliffords, \cite{moussa:2016}]
\label{constantdepthclifford}
 The encoded $\overline{H}$, $\overline{S}$, $\overline{\CNOT}$ gates have constant-depth implementations on the surface code.
\end{fact}

 In particular, an unencoded Clifford circuit can be transformed to an encoded Clifford circuit on the surface code with only constant depth overhead. If we add local stochastic noise to a Clifford circuit, we can propagate the errors to the end of the circuit.

\begin{fact}
	\label{bgktnoise}
	(Propagating noise, Theorems 17 and 23 in \cite{bgkt}) Suppose we have a quantum circuit with noise rate $p$ that creates multiple $\Rec(s) \ket*{\overline{0}}^{\otimes 2}$ states using \Cref{bellprep}. Then suppose that it performs a depth-$D$ Clifford circuit on these states. The state of the system is equivalent to a {\it noiseless} computation with only one layer of local stochastic noise, $E$, following the circuit such that $E \sim \mathcal{N}(O(p^{2^{-O(D)}}))$.
\end{fact}

Suppose that we have finished performing an encoded Clifford circuit on the surface code as \Cref{bgktnoise} describes. We could express the final Pauli error as $X(v)Z(w)$ where $X(v)$ is a Pauli $X$ only on qubits with their corresponding bit in $v$ set to $1$, and similarly for $Z(w)$. Note that for $Z$-basis measurements only the $X(v)$ term will affect the measurement outcomes.  Measuring one logical qubit, we obtain a length-$m$ bitstring $\mathcal{Y}$ encoding the measurement outcome. The decoding function $\Dec \colon \{0,1\}^m \to \{0,1\}$ selects whichever logical measurement is most likely, and is robust to errors in the $m$ physical measurements:

\begin{fact}
	\label{codeblockcorrect}
	(Lemma 21 in \cite{bgkt}) Suppose that there is only one layer of local stochastic noise $X(v) \sim \mathcal{N}(r)$ with $r \leq 0.01$ which occurs immediately before the measurement of any codeblock. Then
	\begin{equation}
		\Pr_{X(v) \sim \mathcal{N}(r)}[\Dec(x \oplus v) = \Dec(x)] \geq 1 - \exp (- \Omega (m^{1/2}))
	\end{equation}
	for any $x \in \mathcal{L}$.
\end{fact}

This fact says that when the layer of local stochastic noise is below a constant threshold of $0.01$, then $\Dec$ will successfully decode the surface code measurement outcome for any $x$ in the codespace.

\subsection{Generic Graph State Measurement Problem}
\label{sec:problem_statement}

All of the problems we consider fit into the following framework. 
\begin{problem}[$k$-Round Graph State Measurement problem] \label{prob:interactive_def}
Let $k \geq 1$ be an integer. Let $\{ G_n = (V_n, E_n) \}_{n \geq 1}$ be a uniform family of graphs, where $G_n$ has $|V_n| = O(\poly(n))$ vertices. Furthermore, for each graph $G_n$, suppose the vertices are colored with $k$ colors, i.e., there exists $\chi_n \colon V_n \to [k]$ for all $n \geq 1$. Each choice of $k$, $\{ G_n \}$, and $\{ \chi_n \}$ defines a problem within this framework. 

Let $\CZ(i,j)$ denote a controlled-$Z$ gate on qubits $i$ and $j$. The problem is to prepare the graph state $\ket{G_n}$, where
$$
\ket{G_n} := \prod_{(i,j) \in E_n} \CZ(i,j) \ket{+}^{\otimes \left| V_n \right|},
$$
and then measure the vertices in $k$ rounds of interaction. In the $i$th round, measurement bases (either $X$ or $Y$) are provided for all vertices of color $i$ (i.e., $\chi^{-1}(i)$), and the device is expected to output corresponding measurement outcomes (either $+1$ or $-1$) for each such vertex. At the end of $k$ rounds, the device succeeds if measuring $\ket{G_n}$ in the input measurement bases could generate (with non-zero probability) the output measurement outcomes.
\end{problem}
Within this paper, we will consider only non-interactive relation problems ($k=1$, following \cite{bgk, bgkt}, etc.), and two-round interactive problems ($k=2$, following \cite{gs}) where the respective graphs are defined in the literature. As a result, we will avoid precisely defining the families of graphs and rely on the following properties of the graphs:
\begin{itemize}
    \item The family of graphs is efficiently constructible and uniform (i.e., all graphs are constructed by the same machine). In other words, we will assume any basic processing of the graph (e.g., enumerating the vertices, determining adjacency, etc.) is not a bottleneck in our complexity reductions. 
    \item The maximum degree of the vertices is $O(1)$ to allow constant-depth ($\QNC^{0}$) construction of $\ket{G_n}$. 
    \item The state $\ket{G_n}$ may be used as a resource for MBQC (measurement-based quantum computation) for some family of quantum circuits that \emph{is} precisely defined. 
\end{itemize}

Additionally, all of the problems that we describe can be solved by families of \emph{classically-controlled} Clifford circuits. A \emph{classically-controlled} Clifford circuit is a Clifford circuit which receives classical input $x \in \{0,1\}^n$ and has the property that each gate in the circuit is uniquely controlled by a single bit of the classical input $x$. That is, the gate indexed by $i$ is activated iff $x_i=1$.

With this definition in hand, we make the following straightforward claim. 

\begin{theorem}
Any $k$-Round Graph State Measurement problem on graphs $\{ G_n \}_{n \geq 1}$ with $O(1)$ maximum degree is solved by a family of $O(1)$-depth, classically-controlled Clifford circuits. Furthermore, the circuits are uniform (i.e., generated by a fixed Turing machine, say, given the input $n$) if the family of graphs is uniform. 
\end{theorem}

The theorem above follows almost by definition of the problem itself.  \Cref{fig:quantum_interactive_task} depicts a circuit for a general $2$-round protocol. In the $2$-Round Graph State Measurement problem there is no garbage (all qubits are eventually measured), but otherwise a circuit for the task looks very much like \Cref{fig:quantum_interactive_task}. 

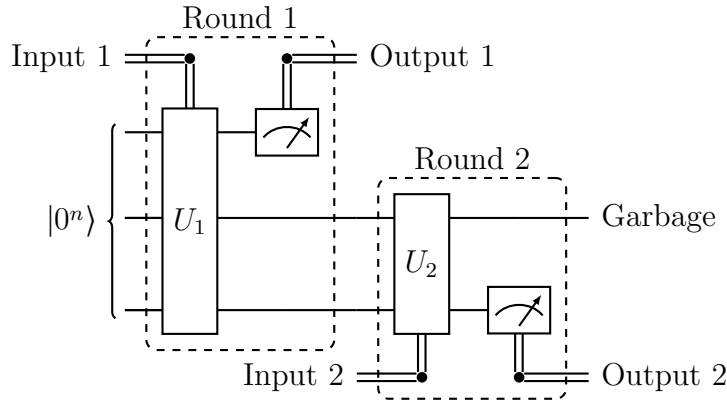
\begin{figure}[ht]
	\centerline{
		\begin{quantikz}
	\lstick{Input 1} & \hackcontrol{}\vcw{1}\cw \gategroup[wires=4,steps=2,style={dashed,inner sep=2pt,rounded corners}]{Round 1}  &  \hackcontrol{}\vcw{1}  & \cw \rstick{Output 1} \\
	\lstick[wires=3]{$\ket{0^n}$} & \gate[3]{U_1} & \meter{} & \\
	\qw & \qw & \qw & \qw & \gate[2]{U_2} \gategroup[wires=3,steps=2,style={dashed,inner sep=2pt,rounded corners}]{Round 2}  &  \qw & \qw \rstick{Garbage} \\
	\qw & \qw & \qw & \qw &  & \meter{} \\
	& & & \lstick{Input 2} & \hackcontrol{} \vcw{-1}\cw & \hackcontrol{} \vcw{-1} & \cw \rstick{Output 2} 
\end{quantikz}
	}
	\caption{General form of a quantum circuit for a $2$-round interactive protocol.}
	\label{fig:quantum_interactive_task}
\end{figure}

\subsection{Relation and Interactive problems}
\label{sec:rel_and_int}

As discussed above, we will only consider problems with one or two rounds of interactivity. A one-round protocol corresponds to what would ordinarily be called a \emph{relation problem}, i.e., given an input $x \in \{ 0, 1 \}^n$, produce an output $y \in \{ 0, 1 \}^{m}$ (where $m$ is polynomially related to $n$) such that $(x, y)$ belongs to some relation $R \subseteq \{ 0, 1 \}^n \times \{ 0, 1 \}^{m}$. In this paper, for all $x \in \{0,1\}^n$ in the promise there exists some $y \in \{ 0, 1 \}^m$ satisfying the relation, but we note that it may not be unique (for us, often \emph{half} the strings satisfy the relation). We use the notation $R(x,y) = 1$ to indicate success (i.e., $(x,y)$ satisfies the relation problem), otherwise $R(x, y) = 0$. 

The two-round tasks are back-and-forth interactions between a question-asker (challenger) and a computational device. The challenger provides the first input $x_1$, and the device replies with its first output $y_1$. The challenger provides a second input, $x_2$, possibly dependent on $y_1$, and the device is expected to reply with some output $y_2$. The sequence of back-and-forth messages $(x_1, y_1, x_2, y_2)$ is known as the \emph{transcript}. The problem itself is defined by a set of transcripts $T$, and the device succeeds if the final transcript belongs to $T$.  

For the purposes of complexity theoretic reductions, we treat the device as an \emph{oracle with state}, in a very natural way. We feed the first round input, it returns the first round output, then the next input it treats as the second round input (i.e., continuing the previous interaction, making use of its internal state), and answers with the second round output. After the last round of input/output, feeding another input to the oracle will cause it to restart the whole protocol again from the beginning. 

\subsubsection{Rewind Oracle}
\label{sec:rewindoracle}

Ordinarily, oracles are treated as black boxes in reductions. However, the results of Grier and Schaeffer \cite{gs} depend on a pivotal difference between a \emph{classical} implementation of an interactive oracle and a \emph{quantum} implementation. Crucially, when a classical device solves an interactive problem, the internal state of the device may be duplicated at any time during the protocol, and later restored to that point. That is, in addition to restarting the interactive protocol afresh, we have the option to \emph{rewind} it to any previous point in the computation, and perhaps continue with different inputs. It turns out that for many classical complexity classes, if the oracle belongs to that class, then so does the oracle with rewinding.

\begin{fact}
Suppose $\mathcal{O}$ is an oracle for an interactive problem, and let $\mathcal{R}$ be the \emph{rewind oracle} obtained by giving rewind capability to $\mathcal{O}$. If $\mathcal{O}$ is implemented by a device in some classical complexity class $\mathcal{C}$, then $\mathcal{R}$ is also in $\mathcal{C}$ as long as $\mathcal{C}$ is one of $\AC^0, \AC^0[p], \TC^0, \NC^1, \cL$.
\end{fact}
\begin{proof}[Proof sketch]
    For the case of a circuit, one can simply fix the outputs of the gates up to some point and use these outputs as additional input to simulate subsequent points of the computation multiple times. For a Turing machine, one copies the contents of the tape up to some point and restarts the TM from these contents to simulate subsequent points of the computation multiple times.
\end{proof}

In the context of $2$-round graph state measurement, the oracle's answers represent measurement outcomes for a round of qubits. With a rewind oracle, we can measure the second round qubits and then rewind back to the end of the first round. We are then free to choose a different basis (i.e., change the second round input) and measure again, as many times as we choose. This is a powerful ability which an actual quantum device does not have since, in general, it is impossible to make two non-commuting measurements on a state, even if the measurements are known in advance. We formally state all of our hardness reductions in terms of the rewind oracle, meaning that the problem is only ($\NC^{1}$, $\parityL$, etc.)-hard for classical devices, since a quantum solution does not imply a rewind oracle. 

Last, although rewind oracles appear to give the power to measure the second round qubits in arbitrary bases, the outcomes are not truly \emph{samples} from the distribution of measurement outcomes. Because of our definition of interactive protocols (and the $k$-round graph state measurement problem), we can only claim that oracle returns a \emph{possible} outcome, not necessarily with the same probability, or perhaps not even random at all! 

In fact, we must assume the oracle produces adversarial outputs, designed to thwart our hardness reductions at every turn. For example, suppose we wish to distinguish between two non-orthogonal states. For any measurement, there is some outcome consistent with both states, and in the worst case the oracle will always give us that outcome. \cite{gs} combats this in two ways: First, they use demonstrations of non-contextuality (specifically, the magic square and magic pentagram games) to force the oracle to reveal some information about the state. Second, they use self-randomization subroutines to conceal the original query within a uniformly random input query, so that the oracle cannot tell which pair of states we want to distinguish, and thus cannot reliably choose an outcome common to both. See \cite{gs} for a full description of this process; we re-use many of their ideas, but make changes in our proof of \Cref{parityLmain} to increase the input query randomization. We discuss the motivation behind this change in the corresponding sections. 

\subsection{Measurement-based Quantum Computation}
\label{sec:mbqc}

Many of the results in \cite{gs} are based on the hardness of computing the final state (or even distinguishing between two possible final states) in some depth-$\Omega(n)$ Clifford circuit. However, we need the simulation to be performed by a constant-depth quantum circuit, so we use the following theorem from measurement-based quantum computation (MBQC) \cite{rb, rbb}:

\begin{theorem} \label{thm:mbqc}
    Fix the \emph{layout} of a circuit of one- and two-qubit gates, i.e., a circuit diagram where all the gates are placeholders to be replaced with a concrete one- or two-qubit unitary later. Let us map each placeholder gate to a \emph{gadget} (details omitted, see \cite{rb, rbb} or \cite{gs}), in the form of a constant-size graph, and connect these gadgets together (as they are connected in the circuit) to form a graph $G$.
    
    For any concrete unitary $U$, there exists a set of measurement bases for the qubits of the corresponding gadget, such that measuring in those bases performs $PU$ for some Pauli operation $P$ depending on the measurement outcomes. That is, if we plug in unitaries to get a circuit $U_k \cdots U_1$ and measure the gadget qubits in the appropriate bases, then the remaining qubits are in state $P_k U_k \cdots P_1 U_1 \ket{+^n}$ for Pauli operations $P_1, \ldots, P_k$. Moreover, if the unitary operations are Clifford then $X$- and $Y$-basis measurements suffice.
\end{theorem}

In other words, we can simulate any Clifford circuit with a particular layout by measuring a constant-degree graph, and thus with a constant-depth circuit. The only catch is that there are Pauli errors throughout the circuit. Fortunately, the Clifford group normalizes the Pauli group by definition, so each Pauli operation can be ``pushed'' through the circuit at the expense of some computation. Unfortunately, pushing all the error to the end of the circuit is effectively as hard as simulating the circuit outright, and therefore not practical in our hardness reductions.

Finally, we note that MBQC works as expected with interactive protocols. In our interactive protocols, for instance, we let nearly all of the gates be in the first round, and only a handful, comprising a constant-depth circuit at the end, into the second round. Then for all the gates in a round, we measure the qubits of their gadgets appropriately, computing the state $\ket{\psi} := P_2 U_2 P_1 U_1 \ket{+}^{n}$ where $U_1, U_2$ are the Clifford unitaries represented by the two parts of the circuits, and $P_1, P_2$ are Pauli operations. Since $U_1$ is a deep Clifford circuit, we are unable to compute $P_1$ in our reductions, but we \emph{can} compute $P_2$ since it is easy to push Pauli operations through a constant-depth circuit such as $U_2$. In the second round we also measure $\ket{\psi}$ in, say, the $X$-basis (this is part of the second round), which means we can tell which outcomes $P_2$ flipped, and translate the actual measurement outcomes to outcomes for the state $U_2 P_1 U_1 \ket{+}^n$. 

Thus, we can think about our interactive protocols like this: the first round input specifies Clifford gates for the placeholders in some layout. The first round output determines $P_1$, but not in an easily decodable way, not unlike a cryptographic commitment. The second round input specifies more gates, but since they have constant depth, they are used only to change the measurement basis. Finally, the second round output tells us, with some constant-depth classical processing, the result of measuring $P_1 U_1 \ket{+}^n$ in the measurement basis specified in the second round input. With a rewind oracle, this gives us the power to measure $P_1 U_1 \ket{+}^n$ in different bases with the \emph{same} unknown $P_1$ each time.

\section{The noisy extension and \texorpdfstring{$\AC^0$}{AC0} separation}

In this section, we review the noisy relaxation of the $1$-Round Graph State Measurement Problem of \cite{bgkt} called the \emph{noisy extension}. We will revisit the main results of that paper, and show how their separation between noisy $\QNC^0$ and $\NC^0$ can be extended to a separation between noisy $\QNC^0$ and $\AC^0$ using the results of Bene Watts et al.\ \cite{bkst}.  The results of this section are not independent from our results for interactive Clifford simulation tasks since the noisy extension will play a critical role there, as well.

\subsection{The noisy extension}

Suppose that we have a relation problem defined by $R \subseteq \{0,1\}^* \times \{0,1\}^*$ that is solved with certainty by a classically-controlled Clifford circuit $C$. On input $x \in \{0,1\}^n$, the circuit is tasked with finding a $y$ such that $(x,y) \in R$, or in other words, it is tasked with solving the reation problem defined by $R$. Let us name this circuit on input $x$ as $C_x$. The computation begins with multiple copies of $\ket{0}$, applies $C_x$, and measures all qubits as output. To make this relation noise-tolerant, we convert $R$ to its {\it noisy-extended} version. The noisy-extended version is defined using the 2D surface code in \Cref{sec:surface_code}. Its effect on a relation problem is the following: For some input $x$, the set of $y$ such that $(x,y) \in R$ is enlarged to the set of $\mathcal{Y}$ that decode to $y$. 

However, recall that the procedure for basis state preparation on the surface code incurs an additional Pauli operator $\Rec(s)$. Fortunately, the effect of this Pauli operator on the overall quantum computation can be propagated through the circuit. Consider the classically-controlled Clifford circuit on input $x$, $\overline{C_x}$. Since the Clifford group normalizes the Pauli group, we can define $f(s,x)$ and $h(s,x)$ by

\begin{equation}
	\label{f}
	X(f)Z(h) \sim \overline{C_x} \Rec(s) \overline{C_x}^{\dagger}
\end{equation}
where $f=f^1 \ldots f^n$ and each $f^i$ is $m$ bits describing the Pauli $X$ operator on the $i$'th codeblock of the circuit. We are now ready to define the noisy-extended relation.

\begin{definition}
	\label{noisyextended}
	The \emph{noisy-extended relation} $R'$ associated with relation $R$ is defined as
	
	\begin{equation}
		R'(x, (\mathcal{Y},s)) = \begin{cases*}
      					1 & if $R(x, y) = 1$ for $y_i = \Dec(\mathcal{Y}^i \oplus f^i(s,x)) \quad \forall \, i \in [n]$ \\
      					0 & otherwise
   				    \end{cases*}
	\end{equation}
\end{definition}

As a sanity check, let us first show that if a {\it noiseless} quantum circuit solves a relation problem $R$ with certainty, then the noiseless quantum circuit encoded by the 2D surface code solves the noisy-extended relation problem $R'$ with certainty. Through the basis state preparation procedure (\Cref{bellprep}), the circuit prepares the state $(\Rec(s^1) \otimes \ldots \otimes \Rec(s^n))\ket*{\overline{0}^n} = \Rec(s)\ket*{\overline{0}^n}$ obtaining syndrome outputs $s = s^1 \ldots s^n$. It then performs a classically-controlled Clifford circuit (on input $x$) $\overline{C_x}$ in constant-depth (\Cref{constantdepthclifford}) and measures the output $\mathcal{Y}$ with the property 
\begin{equation}
	| \bra{\mathcal{Y}} \overline{C_x} \Rec(s) \ket*{\overline{0}^n} |^2 > 0.
\end{equation}
By propagating the Pauli $\Rec(s)$ over Clifford circuits using \Cref{f}, we have that $| \bra{\mathcal{Y}} X(f)Z(h)\overline{C_x}  \ket*{\overline{0}^n} |^2 > 0$. Since Pauli $Z$ operators have no effect on $Z$-basis measurement, we get 
\begin{equation} 
	| \bra{\mathcal{Y} \oplus f(s,x)} \overline{C_x} \ket*{\overline{0}^n} |^2 > 0.
\end{equation}
Finally, because $\overline{C_x}$ is the encoded version of $C_x$, the bit string $y$ defined by 
\begin{equation}
	\label{decoding}
	y_i = \Dec(\mathcal{Y}^i \oplus f^i(s,x)) \; \forall i \in [2n]
\end{equation} 
satisfies the original relation $R$, and therefore the output of the entire procedure (i.e., $(\mathcal Y, s)$) satisfies the noisy extended relation $R'$.  Of course, the whole point of defining the noisy extension is to show that the outputs of a noisy quantum circuits still satisfy it:

\begin{theorem}[Noise-tolerance of noisy extension, Theorem 17 in \cite{bgkt}]
\label{blackboxnoise}
	 Suppose a constant-depth classically-controlled Clifford circuit solves a relation problem. Then, there exists another constant-depth classically-controlled Clifford circuit solving its corresponding noisy extended relation problem with probability at least $1 - \exp(-\Omega (m^{1/2}))$ when the noise rate is below some constant constant threshold.
\end{theorem}

\subsubsection{\texorpdfstring{$\mathsf{AC^0}$}{AC0} decoding circuit}

The noisy extension allows a noisy constant-depth quantum circuit to solve the modified relation problem with high probability.  However, in order to demonstrate a separation between the quantum and classical circuits' capabilities, we need to argue that the noisy extended relation problem remains hard for classical circuits. Our goal is to show that any $\qAC^0$-capable device can actually decode measurements of the physical qubits into measurements of the logical qubits.  Therefore, any such device which can solve the noisy extended relation problem can also solve the original relation problem.  This allows us to port previous noiseless separations to the noisy setting.  

Recall that the encoded quantum circuit will output $\mathcal{\mathcal{Y}}$ and $s$ such that the bitstring $y$ defined by \Cref{decoding} is the output of an unencoded circuit. We show that decoding can be carried out by an $\exp(m, m_{anc})$-size $\mathsf{AC^0}$ circuit.\footnote{In fact, this size is nearly optimal for $\AC^0$ circuits because computing the parity of $n$ bits reduces (using only projections) to $\Dec$ on $O(n^2)$ bits.} The construction is shown in \Cref{decgadget}.

\usetikzlibrary{arrows, automata}

\begin{figure}[ht]
	\centering
    \begin{tikzpicture}[
            > = stealth, 
            shorten > = 1pt, 
            auto,
            node distance = 3cm, 
            semithick 
        ]

        \tikzstyle{every state}=[
            draw = black,
            thick,
            fill = white,
            minimum size = 4mm
        ]

	\node (x) at (-1, 0) {$x$};
	\node (s) at (-1, 1) {$s$};
	\node (Y) at (-1, 2) {$\mathcal{Y}$};
	
	\node[state] (rec) at (1, 1) {$\Rec$};
	\path[->] (s) edge node {} (rec);
	
	\node[state] (f) at (3.5, 0) {$f$};
	\path[->] (rec) edge node[sloped] {$\Rec(s)$} (f);
	\path[->] (x) edge node {} (f);
	
	\node[state] (xor) at (6, 1) {$\oplus$};
	\path[->] (Y) edge node {} (xor);
	\path[->] (f) edge node[sloped, below] {$f^i(s,x)$} (xor);
	
	\node[state] (dec) at (10, 1) {$\Dec$};
	\path[->] (xor) edge node {$\mathcal{Y}^i \oplus f^i(s,x)$} (dec);
	
	\node (y) at (14,1) {};
	\path[->] (dec) edge node {$\Dec(\mathcal{Y}^i \oplus f^i(s,x))$} (y);
	
\end{tikzpicture}
\caption{Decoding gadget} \label{decgadget}
\end{figure}
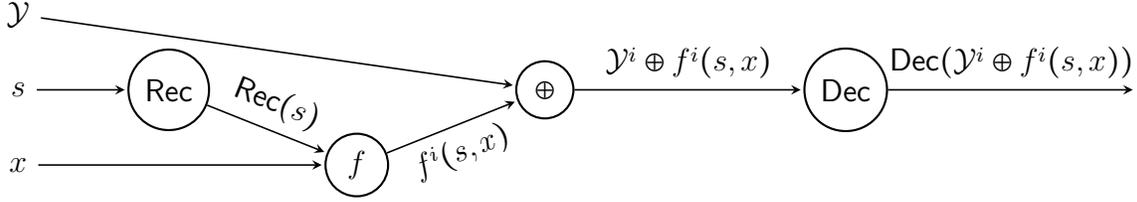

\begin{lemma}
	\label{decode}
	Given $0/1$ bitstrings $\mathcal{Y}$, $s$, and $x$ of lengths $m$, $m_{anc}$, and $n$, respectively, $\Dec(\mathcal{Y}^i \oplus f^i(s,x))$ can be computed by an $\mathsf{AC^0}$ circuit of size $\exp(m_{anc}, m)$ for each $i$.\footnote{We assume that $\poly(n) \leq \exp(m,m_{anc})$ throughout the paper, so we omit the circuit's size dependence on $n$.}
\end{lemma}
\begin{proof}
    Note that $\Rec$ and $\Dec$ can be computed by $\exp(m, m_{anc})$ size circuits by using a truth-table circuit for these functions. We would like to conjugate the Pauli $\Rec(s) = \Rec(s^1) \otimes ... \otimes \Rec(s^n)$ with the classically-controlled Clifford circuit, $\overline{C_x}$, to obtain $f^i(s,x)$ in \Cref{f}. Consider any depth-one Clifford circuit, $C$, on $mn$ qubits composed of one- and two-qubit gates. Conjugating the Pauli operator $\Rec(s)$ by $C$ is locally computable by a polynomial size $\AC^0$ circuit because $C$ only uses one- and two-qubit Clifford gates. Repeating this process for each layer of the constant-depth circuit $\overline{C_x}$, XOR'ing each bit with $\mathcal Y^i$, and plugging the output into $\Dec$ gives the lemma.
\end{proof}

In summary, we have computed $\Dec(\mathcal{Y}^i \oplus f^i(s,x))$ for each $i$ with a $\exp(m, m_{anc})$ size $\AC^0$ circuit.  Because $m,m_{anc} = \polylog(n)$ in all of our constructions, we have shown that the decoding can be done in $\qAC^0$.  This quasipolynomial blowup will not affect the conclusion of this section since neither poly-size nor quasipoly-size $\AC^0$ circuits can solve parity.  However, this blowup \emph{will} be important to our conclusions about interactive problems (see, for instance, \Cref{thm:parityL_noise_sep} in the Introduction).

\subsection{Noise-tolerant \texorpdfstring{$\mathsf{AC^0}$}{AC0} separation}

	
We begin by reviewing the relevant problem that is solvable by {\it noiseless} $\mathsf{QNC^0}$ circuits but is hard for $\mathsf{AC^0}$ circuits to solve. By applying the noisy extension introduced in the previous section, we can prove that a separation persists as the quantum circuit is subject to noise. 

It was shown in \cite{bkst} that there is a problem that is solved with certainty by a $\mathsf{QNC^0}$ circuit but is hard for $\mathsf{AC^0}$ circuits to solve. The problem, a promise version of 1-Round Graph State Measurement Problem that they call the Relaxed Parity Having Problem (RPHP), is a relation problem with inputs uniformly chosen from a set $P_n \subseteq \{0,1\}^n$ for all $n$ that is solved with certainty by a constant-depth classically-controlled Clifford circuit.\footnote{We will actually use \emph{Parallel} RPHP (under the same name RPHP, for simplicity), a version with polynomially many copies of vanilla RPHP, since this gives us better error guarantees.}  RPHP is hard for $\mathsf{AC^0}$ circuits with an error rate exponentially close to $1$. 

\begin{lemma}[Theorem 26 in \cite{bkst}]
\label{repeatrphphard}
Any probabilistic $\mathsf{AC^0}$ circuit of size $s$ and depth $d$ cannot solve RPHP with probability exceeding
\begin{equation}
	\exp(\frac{-n^{1/2-o(1)}}{O(\log s)^{2d}})
\end{equation}
over a uniformly random input $x \in P_n$\footnote{The original statement of this claim in \cite{bkst} refers to deterministic $\AC^0$ circuits instead of probabilistic circuits. A simple argument shows that it also holds for probabilistic circuits: A probabilistic circuit is a convex combination of deterministic circuits, and every deterministic circuit in this convex combination succeeds with probability at most $\exp(-n^{1/2-o(1)}/O(\log s)^{2d})$ over a uniformly random input.}.
\end{lemma}

Define \emph{Noisy} RPHP to be the noisy extended relation problem associated with RPHP (see \Cref{noisyextended}) with $m, m_{anc} = \Theta(\polylog(n))$.  By construction, a noisy quantum circuit can solve Noisy RPHP with high probability.  Plugging this relation problem into \Cref{blackboxnoise}, we get:

\begin{proposition}
	\label{qnoisyparallel}
	There is a constant-depth classically-controlled Clifford circuit that solves Noisy RPHP with probability at least 
	$
		1 - \exp(- \Omega (\polylog (n)))
	$
	over all inputs when the noise rate is below some constant threshold.
\end{proposition}

We conclude with the separation between $\mathsf{AC^0}$ and noisy $\QNC^0$ using Noisy RPHP:

\begin{theorem}
	Any probabilistic $\mathsf{AC^0}$ circuit of size $s$ and depth $d$ cannot solve Noisy RPHP with probability exceeding
	\begin{equation}
		\exp(\frac{-n^{1/2-o(1)}}{O(\log(s + \exp(\polylog (n))))^{2d+O(1)}})
	\end{equation}
	over the uniformly random $x \in P_n$.
\end{theorem}	
\begin{proof}
	We decode the output of Noisy RPHP to one of RPHP using \Cref{decode}. This incurs an extra size overhead of the $\mathsf{AC^0}$ circuit by $\exp(m, m_{anc}) = \exp(\polylog (n))$. By applying this size expansion to \Cref{repeatrphphard}, we arrive at the desired bound.
\end{proof}

\section{Interactive hardness}

To start, let us extend the $2$-Round Graph State Measurement Problem (\Cref{prob:interactive_def}) to its noisy variant by encoding the qubits in the surface code in the natural way:

\begin{problem}[Noisy $k$-Round Graph State Measurement] \label{prob:noisy_interactive_def}
Consider the $k$-Round Graph State Measurement Problem where each logical qubit is encoded in $\Theta(\polylog(n))$ physical qubits according to the surface code (see \Cref{sec:surface_code}).  That is, we start with some encoded graph state
$$
\ket{\overline{G_n(s)}} := \prod_{(i,j) \in E_n} \overline{\CZ}(i,j)  \Rec(s)\ket{\overline{+}}^{\otimes \left|V_n\right|},
$$
where $s$ denotes the syndrome qubits associated with preparing the $|V_n|$ logical $\ket{\overline{+}}$ states. The choice of $s$ is left to the device.  In the $i$th round of interaction, logical measurement bases (either $\overline{X}$ or $\overline{Y}$) are provided for all vertices of color $i$, and the device is expected to output the corresponding measurement outcomes for every physical qubit representing each such vertex. At the end of the $k$ rounds, the device also outputs $s$.

The device succeeds if its outputs match possible outcomes (i.e., literally any outcome with non-zero probability) from measuring $\ket{\overline{G_n(s)}}$ in the bases given as input. This acceptance criterion is equivalent to requiring that the measurements of the encoded graph state satisfies the noisy-extended relation (see \Cref{noisyextended}).
\end{problem}

Showing that there exists a noisy, constant-depth circuit which solves the 2-round problem for constant-degree graphs is relatively straightforward given the properties of the code.

\begin{theorem}
\label{interactiveblackboxnoise}
     There exists a constant-depth classically-controlled Clifford circuit that solves Noisy $2$-Round Graph State Measurement on constant-degree graphs with probability at least $1 - \exp(- \Omega(\polylog (n)))$ for any input when the noise rate is below some constant threshold.
\end{theorem}
\begin{proof}
	We will construct the natural classically-controlled Clifford circuit (see \Cref{interactivecircuit}) using the surface code with logical qubits of size $m = \Omega(\polylog (n))$.  In fact, this circuit is identical to that used for the associated relation problem (i.e., Noisy 1-Round Graph State Measurement), except the second-round inputs are only provided after the first-round measurements are returned by the circuit.  Therefore, the round-one outputs may leak information about the error on the round-two qubits.  However, we will show that the challenger's questions in round two cannot concentrate the error on any given logical qubit and therefore, the error can be decoded normally.
	
	To recap the error correcting components, let us describe the circuit $\mathcal C$:  first, the circuit prepares the logical zero state  $\Rec(s) \ket*{\overline{0}}^{\otimes |V_n|}$ using \Cref{bellprep}; next, the circuit constructs the graph state using a constant-depth circuit $\overline{\mathcal{C}_G}$ using \Cref{constantdepthclifford}; finally, in each round of the protocol, the circuit performs a basis change Clifford operation ($\overline{\mathcal{C}_A}$ and $\overline{\mathcal{C}_B}$ for the first round and second round, respectively) and measures in the $Z$ basis. 
	
	Let us analyze the effect of noise in this process. Consider the error $E$ immediately before the first round measurement.  By \Cref{bgktnoise}, we know that if the noise rate is below some constant threshold, then $E \sim \mathcal N(p)$ with $p < .01$.  In fact, by the definition of local stochastic noise, the error $E = E_1 \otimes E_2$ can be decomposed into two errors (on the round-one qubits and round-two qubits, respectively) both of which are themselves locally stochastic with the same parameter. That is, $E_1 \sim \mathcal N(p)$ and $E_2 \sim \mathcal N(p)$.  Because the error on the first-round qubits $E_1$ has noise parameter $p < .01$, the first round measurements are correct with with probability at least $1 - \exp(- \Omega (\polylog (n)))$ by \Cref{codeblockcorrect}.
	
	By the same token, we also have that $E_2$ is local stochastic noise.  The second-round inputs may depend on $E_2$, but critically, $\overline{\mathcal{C}_B}$ is simply a single layer of logical single-qubit gates.  By \Cref{constantdepthclifford}, we have that each such operation can be applied with a constant-depth circuit on a single codeblock, and so the number of non-identity Pauli elements in $E_2$ can increase by at most a constant factor.\footnote{In fact, the specific construction of the single-qubit logical operations \cite{moussa:2016}  increases the non-identity Pauli elements by at most a factor of 2.}  Furthermore, the rest of the error in the circuit can once again be captured as local stochastic noise right before the second-round measurement by \Cref{bgktnoise}.  Therefore, there is once again a small enough parameter for the noise rate such that the number of errors from both sources remain below the distance\footnote{Although we have mostly referred to the surface code as being able to correct for some specific noise rate, we observe that a low noise rate corresponds to few errors with high probability, and it is the paucity of errors that ultimately determines the success of the code \cite{bgkt}.} of the code with high probability.
\end{proof}

\begin{figure}[ht]
	\centerline{
		\begin{quantikz}
		\lstick{$x_1$} & \cw & \cw & \cw & \cwbend{2} \gategroup[wires=3,steps=4,style={dashed,inner sep=6pt,rounded corners}]{Round 1} & & & & & & & \\
		\lstick{$x_2$} & \cw & \cw & \cw & \cw & \cw & \cw & \cw & \cw & \cwbend{2}  \gategroup[wires=3,steps=4,style={dashed,inner sep=6pt,rounded corners}]{Round 2} & & & \\
		\lstick[wires=2]{$\Rec(s) \ket{\overline{0}}^{\otimes |V_n|}$} & \qw \qwbundle{} & \gate[2]{\overline{\mathcal C_G}} & \qw  \qwbundle{} & \gate{\overline{\mathcal C_A}} \vcw{-2} & \meter{} & \cw \rstick{$\mathcal{Y}_1$} \\
		\qw & \qw \qwbundle{} & & \qw \qwbundle{} & \qw & \qw & \qw & \qw & \qw & \gate{\overline{\mathcal C_B}}\vcw{-2} & \meter{} & \cw \rstick{$\mathcal{Y}_2$} & 
\end{quantikz}
	}
	\caption{Circuit solving Noisy $2$-Round Graph State Measurement Problem}
	\label{interactivecircuit}
\end{figure}
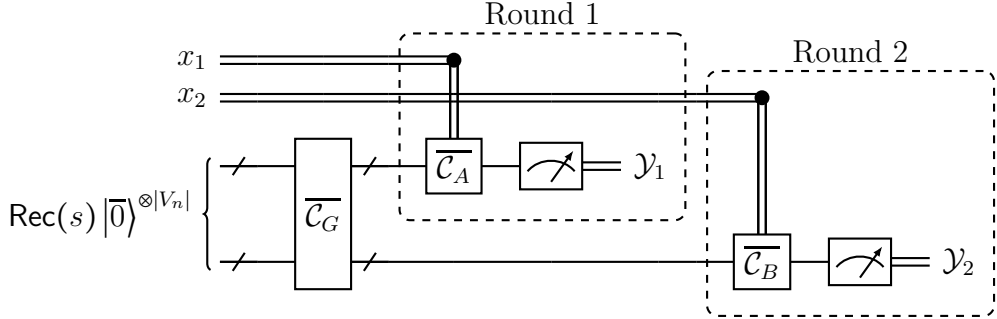

\subsection{\texorpdfstring{$\NC^1$}{NC\textonesuperior}-hardness}
\label{sec:nc1hardness}

This section focuses on the $2$-Round Graph State Measurement problem on a $2 \times O(n)$ grid, especially in the noisy setting.\footnote{The noiseless version of the problem is referred to as the CliffSim[2] in \cite{gs}.  Technically, the graph used in \cite{gs} is a subset of the $2 \times O(n)$ grid arranged in a brick-like pattern. However, such a pattern was only used to reduce the number of qubits per Clifford operation, and the result is qualitatively the same with a simple grid.}  Since the grid is of width 2, the first round measurements are suitable for an MBQC simulation of a sequence of $2$-qubit Clifford gates, and the second round measurements are used to apply one more Clifford gate to permit arbitrary Clifford measurements of a $2$-qubit state.  We will show that a classical simulation of this problem requires solving $\NC^1$-hard problems (\Cref{nc1hardness}), and as a result obtain an unconditional separation between $\AC^0[p]$ circuits and noisy $\QNC^0$ circuits (\Cref{cor:ac0[p]_separation}).

Let us quickly recap how \cite{gs} connects this problem to $\NC^{1}$-hard computation. First, a result of Barrington and Th\'{e}rien shows that computing products in non-solvable groups is $\NC^1$-hard \cite{bt}, even to distinguish between the product being the identity or some other predetermined group element. The $2$-qubit Cliffords modulo $2$-qubit Pauli operations, which we denote $\{\mathcal{C}_{2} / \mathcal{P}_2\}$, turn out to be a non-solvable group (specifically, isomorphic to $S_{6}$). For this special case, Barrington and Th\'{e}rien's result translates to the following:

\begin{theorem}
    \label{thm:nc1hardcliffordproduct}
    Let $C_1, \ldots, C_n \in \mathcal{C}_2 / \mathcal{P}_2$ be given. Promised that the product $C_1 \cdots C_n$ is either $I \otimes I$ or $H \otimes H$ modulo Pauli operations, deciding whether $\ket{\psi} := C_1 \cdots C_n \ket{++}$ is an $X$-basis state (i.e., $\ket{++}, \ket{+-}, \ket{-+}, \ket{--}$) or a $Z$-basis state ($\ket{00}$, $\ket{01}$, $\ket{10}$, $\ket{11}$) is $\NC^{1}$-hard.
\end{theorem}

Our problem for this section is explicitly designed to prepare $\ket{\psi}$ (up to an unknown Pauli $P \in \mathcal{P}_2$) and measure it in an arbitrary Clifford basis. Our goal is to show that a rewind oracle for the problem solves an $\NC^{1}$-hard problem. All that remains is to show how repeated measurements of $P \ket{\psi}$ in arbitrary Clifford bases can be used to determine something about the state, and ultimately leveraged to distinguish $X$-basis states from $Z$-basis states. 

We borrow a tool of \cite{gs} to extract some information about the first round state from the rewind oracle $\mathcal{R}$. We simply quote the lemma here without proof, but it uses a demonstration of non-contextuality, the magic square game, to force the oracle to reveal something about the state, and then it uses a plethora of randomization tricks so the oracle does not control what it reveals. 

\begin{lemma}
	\label{selfreduce} (Theorem 24 in \cite{gs})
	Let $\mathcal{R}$ be a rewind oracle solving the $2 \times O(n)$ interactive problem, possibly with some error (i.e., it may fail the task some fraction of the time). Suppose we are given gates $C_1, \ldots, C_n \in \mathcal{C}_2$. There is a randomized $(\AC^{0})^{\mathcal{R}}$ circuit which 
	\begin{itemize}
	    \item processes $C_1, \ldots, C_n$ to a related but uniformly random sequence of operations,
	    \item calls $\mathcal{R}$ with the above uniformly random sequence in the first round, and makes $6$ uniformly random (but not independent) measurements in the second round (i.e., rewinds 5 times), then
	    \item post-processes the outputs from $\mathcal{R}$ in conjunction with the randomization.
	\end{itemize}
	The result is an algorithm which outputs either a uniformly random stabilizer Pauli (excluding $II$) or uniformly random non-stabilizer Pauli of $C_1 \cdots C_n \ket{++}$. Moreover, the circuit outputs a stabilizer only if at least one call to $\mathcal{R}$ was in error. 
\end{lemma}

We apply this lemma repeatedly to learn many non-stabilizers of $C_1 \cdots C_n \ket{++}$, until we accumulate enough non-stabilizers to determine the state $C_1 \cdots C_n \ket{++}$ up to Pauli operations. Thus, a rewind oracle for the 2-Round Graph State Measurement Problem on a $2 \times \poly(n)$ grid, with randomized $\AC^{0}$ reduction circuits, can solve the $\NC^{1}$-hard problem in Theorem~\ref{thm:nc1hardcliffordproduct}. We do not state the theorem (or proof) since it will generalize it to the noisy version of the problem in \Cref{etolerance}, with proof in Appendix~\ref{improvednc1}.

\subsubsection{Noisy \texorpdfstring{$\NC^1$}{NC\textonesuperior} separation}

Now let us add an error correcting code with $\polylog(n)$ physical qubits per logical qubit, so that a noisy constant-depth quantum circuit can solve the problem with high probability. That is, we extend the original 2-Round Graph State Measurement Problem (on the $2 \otimes \poly(n)$ grid) to its noisy-extended version using \Cref{prob:noisy_interactive_def}. 

By applying \Cref{interactiveblackboxnoise}, we immediately get a constant-depth quantum circuit that solves the problem. 

\begin{cor}
     There exists a constant-depth classically-controlled Clifford circuit that solves the Noisy $2$-Round Graph State Measurement Problem on a $2 \times \poly(n)$ grid with probability at least $1 - \exp(- \Omega(\polylog (n)))$ when the noise rate is below some constant threshold.
\end{cor}

We now prove classical hardness for the same task. We first show that a $(\mathsf{BPAC}^0)^{\mathcal{R}}$ circuit can decide a $\NC^1$-hard problem even if the rewind oracle $\mathcal{R}$ fails on $\frac{1}{30}$ possible inputs. This is a slight improvement over the $\frac{2}{75}$ present in \cite{gs}. Then we reduce the from the noisy problem to the noiseless problem to complete the reduction. 

\begin{lemma}
	\label{etolerance}
	Let $\mathcal{R}$ be the rewind oracle for the $2$-Round Graph State Measurement Problem on a $2 \times \poly(n)$ grid. Suppose for a uniformly random input (first \& second round), $\mathcal{R}$ is incorrect with probability $\epsilon < \frac{1}{30}$. Then
\begin{equation}
	\NC^1 \subseteq (\mathsf{BPAC}^0)^{\mathcal{R}}
\end{equation}
\end{lemma}

We defer the proof to Appendix~\ref{improvednc1}. Clearly, replacing a rewind oracle for the noiseless problem with one for the noisy problem (having the same error rate) makes no difference, except that we require quasipoly-size $\AC^{0}$ for decoding. By decoding its output using \Cref{decode}, we conclude that a $(\qBPAC^0)^{\mathcal{R}}$ solves an $\NC^{1}$-hard problem.

\begin{theorem}
	\label{nc1hardness}
	Let $\mathcal{R}$ be the rewind oracle for the Noisy 2-Round Graph State Measurement Problem on a $2 \times \poly(n)$ grid. Suppose for a uniformly random input (first \& second round), $\mathcal{R}$ is incorrect with probability $\epsilon < \frac{1}{30}$. Then
\begin{equation}
	\NC^1 \subseteq (\mathsf{qBPAC}^0)^{\mathcal{R}}
\end{equation}
\end{theorem}

Quasi-polynomial size bounded-probability $\AC^{0}$ circuits ($\mathsf{qBPAC^0}$) may seem like an unnatural class, but we can relate this result to more familiar classes.

\begin{cor}
\label{cor:ac0[p]_separation}
	For any prime $p$, any probabilistic $\AC^0[p]$ circuit of depth $d$ and size $\exp(n^{1/2d})$ cannot pass the Noisy 2-Round Graph State Measurement Problem on a $2 \times \poly(n)$ grid with probability at least $\frac{29}{30}$ over a uniformly random input.
\end{cor}

\begin{proof}
Suppose that an $\AC^0[p]$ circuit could pass the interactive protocol with probability $\frac{29}{30}$ over the uniform input distribution, implying a rewind oracle of similar size and depth that succeeds with probability at least $\frac{29}{30}$. By the previous results, this gives a $\BPAC^0[p]$ circuit of similar size and depth solving an $\NC^1$-hard problem, which combined with the fact that $\BPAC^0[p]$ is contained in non-uniform $\AC^0[p]$ with similar size and depth \cite{ab}, we arrive at the stated size lower bound by Razborov and Smolensky \cite{razborov, smolensky}.
\end{proof}

\subsection{\texorpdfstring{$\mathsf{\oplus L}$}{ParityL}-hardness}
\label{sec:parityL_hardness}

The main theorem of this section is an average-case $\mathsf{\oplus L}$-hardness result for the $2$-Round Graph State Measurement Problem on a $\poly(n) \times \poly(n)$ grid state, which leads to a conditional separation between noisy $\QNC^0$ circuits and classical log-space machines. We will do this by showing a worst-to-average-case reduction from a simple $\parityL$-hard problem to the $2$-Round Graph State Measurement Problem on a $\poly(n) \times \poly(n)$ grid state.

First, let us review this $\parityL$-hard problem, which concerns the parity of paths in a directed acyclic graph.  Define $M$ to be the set of upper-triangular adjacency matrices of monotone graphs---i.e., a directed graph on vertices $V=\{1, \ldots, n\}$ with  no edges from $j$ to $i$ for $j \geq i$.

\begin{lemma}
	\label{dagparity}
	(\cite{damm}) Let $\DAGParity$ be the problem of deciding the parity of paths from $1$ to $n$ in $A \in M$. Then
	\begin{equation*}
		\parityL \subseteq (\NC^0)^{\DAGParity}.
	\end{equation*}
\end{lemma}

For a detailed description of how an instance of $\DAGParity$ is reduced to an instance of the $2$-Round Graph State Measurement Problem, we refer the reader to \cite{gs}.\footnote{Careful readers might notice that the $\parityL$-hard instances in \cite{gs} are derived from \emph{layered} directed acyclic graphs with directed edges only passing from one layer to the next.  This layered structure in the associated problem $\LDAGParity$ is used throughout those reductions.  Nevertheless, we will show how to convert our $\DAGParity$ instances into instances of $\LDAGParity$ and so the entire reduction goes through as normal.}  Our goal in this section will be to show that randomizing an instance of the $\DAGParity$ problem suffices for an average-case hardness result for the graph state measurement problem.  It will be useful to describe this randomized algorithm as a composition of three algorithms, $\gamma$, $\EAlg$, and $\DAlg$, which once composed together take $A \in M$ as input and output a sequence of operations\footnote{Formally, each one of these operations consists of a CNOT gate followed by a circuit of CZ and Phase gates.} which serve as input to the graph state measurement problem through measurement-based computation gadgets. We will describe $\gamma$, $\EAlg$, and $\DAlg$ in greater detail in the next section.

To obtain our average-case hardness result, we constrain the first round inputs of the graph state measurement problem to be from the set $\gamma \circ \EAlg \circ \DAlg(M)$. We constrain the second round inputs to be from a small set of Pauli measurements $I_2$, defined formally in \Cref{erroranalysisappendix}.
 
\begin{theorem}
\label{parityLmain}
Let $\mathcal{R}$ be the rewind oracle for 2-Round Graph State Measurement Problem on a $\poly(n) \times \poly(n)$ grid promised that input is from the image of $\gamma \circ \EAlg \circ \DAlg(M)$ in the first round and $I_2$ in the second round. If the oracle fails with probability\footnote{The probability is taken over the randomness of the input, as well as the internal randomness of $\mathcal{R}$, since we allow $\mathcal{R}$ answer randomly.} at most $\epsilon < \frac{1}{421}$ on a uniformly random input (first and second round) satisfying the promise, then 
\begin{equation*}
\mathsf{\parity L} \subseteq (\mathsf{BPAC^0})^{\mathcal{R}}.
\end{equation*}
\end{theorem}

Appealing to \Cref{decode}, we obtain the following corollary for the Noisy 2-Round Graph State Measurement Problem: 
\begin{cor}
\label{cor:noisy_rewind_goodness}
If the rewind oracle $\mathcal{R}'$ fails with probability $\epsilon < \frac{1}{421}$ on uniformly random inputs of the Noisy 2-Round Graph State Measurement Problem (under the same promise as \Cref{parityLmain}), then
\begin{equation*}
	\mathsf{\parity L} \subseteq (\mathsf{qBPAC^0})^{\mathcal{R}'}.
\end{equation*}
\end{cor}

The corollary above suffices to prove the conditional separation mentioned in the Introduction (\Cref{thm:parityL_noise_sep}): assuming $\parityL \not\subseteq (\qBPAC^0)^{\cL}$, any probabilistic $\cL$ machine fails the Noisy $2$-Round Graph State Measurement Problem with some constant probability over a uniformly random input.  Suppose that some probabilistic $\cL$ machine fails with probability $o(1)$ over uniform input.  Then by \Cref{cor:noisy_rewind_goodness}, we have $\parityL \subseteq (\mathsf{qBPAC^0})^{\cL}$, contradicting the assumption.

\subsubsection{Input half-randomization and \texorpdfstring{$\gamma \circ \EAlg \circ \DAlg$}{gamma compose E compose D} description}

Before detailing each of the three algorithms, let us briefly describe the purpose of each.  The algorithm $\gamma$ captures the randomization scheme described in \cite{gs}, whose main goal is to randomize the \emph{second} round of the interactive protocol. The result is that an adversarial oracle cannot prevent us from learning the state $\ket{\psi}$ created by the first round measurements. We keep $\gamma$ so we do not have to re-invent the wheel for the second round. However, $\gamma$ does not sufficiently randomize the first round to conceal the instance of the $\parityL$-hard problem we are trying to solve. It is conceivable that the oracle has some small probability of error, but coincidentally concentrates that error on the instance (or small number of instances which $\gamma$ randomizes to the same distribution) that we wish to solve. 

The other procedure, $\DAlg$, is adapted from work of Applebaum, Ishai and Kushilevitz \cite{aik}. They introduce the notion of a \emph{randomized encoding}, where it is possible to transform $x$ to a probability distribution that hides everything about $x$ except some function $f(x)$. For example, and especially relevant to us, they show how a binary matrix of a certain form can be mapped to a uniformly random matrix of the same form, and same determinant modulo $2$. In this case, the randomized encoding hides everything about the input (which, recall, is essentially the problem instance) except for one bit, a property we call \emph{half randomizing}. It turns out the parity of the determinant is even $\parityL$-complete, but since our quantum gates are inherently reversible (i.e., non-zero determinant), we must find a more nuanced way to randomize a $\parityL$-complete problem to adequately randomize the first round. 

Finally, $\EAlg$ connects the two procedures by performing a reduction from the output of $\DAlg$ to the input for $\gamma$.

To describe this approach in more detail, let us recall two additional problems from \cite{gs}: $\LDAGParity$ (the layered variant of $\DAGParity$) is the problem of deciding whether a layered DAG has an even number of paths from a fixed start vertex to a fixed end vertex; $\CNOTMult$ is the problem of deciding whether a sequence of CNOT gates (where each gate acts on two of polynomially many input wires) multiplies to the identity, promised that they multiply to either the identity or a fixed 3-cycle.  Let us also borrow the notation $\CNOT_m$ to denote the set of 2-input CNOT gates that can act on $m$ wires. 

It was shown \cite{gs} that $\LDAGParity$ is $\NC^0$-reducible to $\CNOTMult$ using the deterministic algorithm $\EAlg$.
\begin{lemma}[\cite{gs}]
	\label{Fselfreduction}
	There exists a deterministic algorithm $\EAlg$ in $\NC^0$ that takes as input an adjacency matrix $A \in \{0,1\}^{n \times n}$ of a layered DAG and outputs a sequence of $\CNOT$ gates $g_1,...,g_{\poly(n)} \in \CNOT_{\poly(n)}$ such that 
	\begin{enumerate}
	    \item $A \in \LDAGParity$ iff $(g_1,...,g_{\poly(n)}) \in \CNOTMult$.
	    \item $\EAlg$ is injective.
	\end{enumerate}
\end{lemma}

Then the output of $\EAlg$, $(g_1,...,g_{\poly(n)})$, is converted to first round input for the 2-Round Graph State Measurement Problem using the randomized algorithm $\gamma$. We prove the following crucial property about $\gamma$ in \Cref{appendixparitylapplications}.

\begin{lemma}
	\label{fproperties}
	For any distinct sequences of $\CNOT$ gates $(a_1,...,a_{\poly(n)}) \neq (b_1,...,b_{\poly_n})$, $\gamma(a_1,...,a_{\poly(n)})$ and $\gamma(b_1,...,b_{\poly(n)})$ are uniform distributions over disjoint sets of the same cardinality.
\end{lemma}

The output of $\gamma(g_1,...,g_n)$ is a random sequence of operations; each operation is a $\CNOT$ gate followed by a circuit of CZ and Phase gates. We describe precisely how this sequence is generated in \Cref{appendixparitylapplications}. We note that although the composition $\gamma \circ \EAlg$ is a randomized algorithm converting an instance of $\LDAGParity$ to first round input to the 2-Round Graph State Measurement Problem, its output does not hide the particular instance of $\LDAGParity$ it was given as input.  To obtain this property, we compose $\gamma \circ \EAlg$ with the randomized algorithm $\DAlg$ which employs the random self-reducibility of $\DAGParity$. Altogether, the image of $\gamma \circ \EAlg \circ \DAlg$ is the set of first round inputs to the rewind oracle. We will use $r_{\gamma}$ and $r_d$ as notation for the random bits of $\gamma$ and $\DAlg$ when disambiguation is required.  To prove that $\gamma \circ \EAlg \circ \DAlg$ conceals the first-round input while simultaneously preserving the ability to decided membership in $\DAGParity$, we define a property called {\it half-randomizing}.

\begin{definition}
\label{half-randomizing}

Let $\Pi = (\Pi_{yes}, \Pi_{no})$ be a promise problem. A randomized algorithm $f$ over the promise $\Pi$ is \emph{half-randomizing} if there are two disjoint, equal-sized sets $\mathcal C_{yes}$ and $\mathcal C_{no}$, such that for $i \in \{yes, no\}$
\begin{equation*}
f(x) \equiv U_{\mathcal C_i} \text{ for all } x \in \Pi_{i},
\end{equation*}
where $U_{\mathcal C_i}$ is the uniform distribution over $\mathcal C_i$.

\end{definition}

Functions which are half-randomizing will be useful to us because they can spread out the probability of a few problematic instances relatively evenly amongst all the inputs:

\begin{observation}
\label{2eps}
Let $f$ be half-randomizing. If a uniformly random element of the image $\mathcal{C}_{yes} \cup \mathcal{C}_{no}$ satisfies some property $P$ with probability $\leq \epsilon$, then for any $x \in \Pi_{yes} \cup \Pi_{no}$,  
$$
\Pr[f(x) \text{ has property } P] \leq 2 \epsilon.
$$
\end{observation}

To be concrete, consider a probabilistic oracle that fails with probability $\epsilon$ over a uniformly random input in $\mathcal{C}_{yes} \cup \mathcal{C}_{no}$, and suppose there is a reduction from $\Pi_{yes} \cup \Pi_{no}$ to $\mathcal{C}_{yes} \cup \mathcal{C}_{no}$. An ordinary reduction gives us only average-case guarantees---many inputs in $\Pi_{yes} \cup \Pi_{no}$ map to something the oracle can solve with high probability, but any particular $x \in \Pi_{yes} \cup \Pi_{no}$ may or may not be solved with high probability. If the reduction $f$ is half-randomizing, then $x$ maps to a uniform distribution over $\mathcal{C}_{yes}$ or $\mathcal{C}_{no}$, and can be solved with probability at least $1 - 2 \epsilon$, over the randomness of $f$ and the probabilistic oracle, by the observation.

We now describe the half-randomizing algorithm $\DAlg$, which maps an instance of $\DAGParity$ to $\LDAGParity$.  It is composed of two separate subroutines: first, the input (an adjacency matrix for a DAG) is mapped to a uniformly random input sharing the same parity of paths from source to target; second, we convert the DAG to a layered DAG.  In the first subroutine, we use randomized encoding techniques introduced by \cite{aik}:

\begin{fact}
\label{det}
(Fact 4.13, \cite{aik}) Let $A$ be an $n \times n$ monotone adjacency matrix, and let $L$ be the $(n-1) \times (n-1)$ top-right submatrix of $A-I$. Then $\det(L) \mod 2$ is the parity of the number of paths from $1$ to $n$ in $A$. 
\end{fact}

Since $\DAGParity$ is $\parityL$-hard, this fact shows that the determinant of $L$ encodes the answer to a $\mathsf{\parity L}$-hard problem. Note that $L$ is upper-triangular {\it except} for $-1$ on its second diagonal. Furthermore, the randomized encoding scheme of Applebaum, Ishai and Kushilevitz allows us to sample uniformly from matrices that have the same form and same determinant:

\begin{proposition}[Lemma~4.17, \cite{aik}]
\label{sample}
 Given matrix $L$, a randomized $\NC^0$ circuit can sample polynomially many Boolean variables
\begin{equation*}
	\{(K_{i,j}^{(1)}, K_{i,j}^{(2)}, \ldots )\}_{1 \leq i \leq j \leq n-1}
\end{equation*}
uniformly such that the $(n-1) \times (n-1)$ matrix $K^{\oplus}$ defined by $K_{i,j}^{\oplus} = \oplus_{\ell} K_{i,j}^{(\ell)}$, $-1$ on its second diagonal, and $0$ below it has\footnote{The construction of $\{(K_{i,j}^{(1)}, K_{i,j}^{(2)}, \ldots)\}_{1 \leq i \leq j \leq n-1}$ proceeds by sampling random upper triangular matrices with a particular structure, $R_1$ and $R_2$, to produce $K^{\oplus} = R_1LR_2$. Each element in the upper triangular region of $K^{\oplus}$ can be described by a sum of degree-3 monomials whose values after some further randomization, $\{(K_{i,j}^{(1)}, K_{i,j}^{(2)}, \ldots)\}_{1 \leq i \leq j \leq n-1}$, are computable in $\NC^0$. For full details, see \cite{aik}.}
\begin{equation}
\label{detconstraint}
	\det(K^{\oplus}) \equiv \det(L) \mod 2.
\end{equation}
\end{proposition}

Recall that $L$ is derived from our original monotone adjacency matrix $A$.  The lemma above shows that we can generate a uniformly random matrix $K^{\oplus}$ that has the same form as $L$, but we now must convert that matrix back into an adjacency matrix in order to use it as input to $\DAGParity$.  We see that similar to how $L$ is a submatrix of $A-I$, $K^{\oplus}$ is also the submatrix of another matrix $B-I$. As shown in \Cref{fig:B_and_k}, we define $B \in \{0,1\}^{n \times n}$ as the matrix whose entries above the main diagonal are the same as those entries on or above the main diagonal of $K^{\oplus}$, and whose entries on or below the main diagonal are $0$.

\begin{figure}[ht]
	\centering
	\scalebox{.8}{
\begin{tikzpicture}[
Matrix/.style={
	matrix of nodes,
	text height=2.5ex,
	text depth=0.75ex,
	text width=3.25ex,
	align=center,
	left delimiter={[},
	right delimiter={]},
	column sep=0pt,
	row sep=0pt,
	nodes in empty cells,
},
DA/.style={
	fill=red!30, draw=none,
	rounded corners=5pt,
}
]

\matrix[Matrix] at (0,0) (BI){ 
	1 & &  &  &  & \\
	0 & 1 & & & & \\
	0 & 0 & 1 & & & \\
	0 & 0 & 0 & 1 & & \\
	0 & 0 & 0 & 0 & 1 & \\
	0 & 0 & 0 & 0 & 0 & 1 \\
};
\begin{scope}[on background layer]
\draw[DA]
(BI-1-2.north west) 
-| (BI-5-6.south east)
-- (BI-5-6.south west)
|- (BI-4-5.south west)
|- (BI-3-4.south west)
|- (BI-2-3.south west)
|- (BI-1-2.south west)
-- cycle; 

\draw[draw=black!30, rounded corners=5pt] 
(BI-1-2.north west) 
-| (BI-5-6.south east) 
-| (BI-1-2.south west)
-- cycle;        
\end{scope}

\node[above=2em of BI,xshift=-3cm] (K) {\large $K^{\oplus}$};
\node[left=1em of BI] {\Large $B + I = $};

\draw[black!30] (K) -> (BI-1-2.north);

\node[fill=red!30, draw=none, rounded corners=5pt, inner sep=1em, right=3em of K] (blob) {};
\node[right=0.2em of blob] {$= \text{random s.t. $\det(K^{\oplus}) \equiv \det(L) \pmod{2}$}$};

\end{tikzpicture}
}
	\caption{Relationship between $B$ and $K^{\oplus}$}
	\label{fig:B_and_k}
\end{figure}
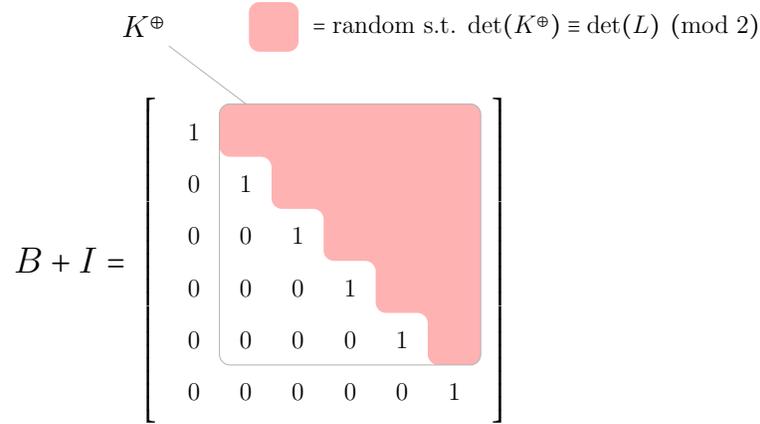

\begin{proposition}
\label{prop:B}
	$B$ is a uniformly random monotone adjacency matrix such that the parity of paths from $1$ to $n$ is the same as in $A$.
\end{proposition}
\begin{proof}
	By construction $B$ is a monotone adjacency matrix, and $K^{\oplus}$ is the $(n-1) \times (n-1)$ top-right submatrix of $B-I$. The entries on or above the main diagonal of $K^{\oplus}$ are uniformly random subject to the determinant constraint in \Cref{detconstraint}, and therefore, so are the entries above the main diagonal of $B-I$ (and hence $B$). However, by \Cref{det}, this constraint is equivalent to $A$ and $B$ having the same parity of paths from $1$ to $n$.
\end{proof}

\Cref{prop:B} completes the first subroutine of $\DAlg$---that is, $B$ is a uniformly random monotone adjacent matrix which has the same parity of paths from the source to target as $A$.  We now turn to the second half of the algorithm, which is converting $B$ (and instance of $\DAGParity$) into a layered monotone adjacency matrix (and instance of $\LDAGParity$) that has the same parity of paths from source to target. To do this, we give a simple deterministic procedure which takes $B$ (actually the entries $\{(K_{i,j}^{(1)}, K_{i,j}^{(2)}, \ldots)\}_{1 \leq i \leq j \leq n-1}$ from which it is comprised) and creates a new layered DAG $C$ with $2n - 1$ layers.

Consider first a procedure for where all the entries of $B = \{b_{i,j}\}$ are given explicitly (recall that in our case, each entry of $B$ is a sum of $K_{i,j}^{(\ell)}$ terms).  To construct $C$, we create $n$ layers, each of which is comprised of a copy of all $n$ nodes in the graph.  Between layers $i$ and $i+1$, there are two types of edges: first, there is a path from each node $k$ to the same node $k$; second, there is a path from node $i$ to node $j$ for each $j$ such that $b_{i,j} = 1$. One can show then at layer $k$, the number paths you can take from node $1$ to any other node in the graph is the number of paths you can take by using edges $b_{i,j}$ with $i < k$, completing the construction.  To handle the fact that the entries of $B$ are actually \emph{sums} of $\NC^0$-computable values, we construct additional paths between the layers for each such value (which requires adding an additional layers to prevent multiple edges).

Let us now formally construct the matrix $C$ that we sketched above.  First, we have $n$ layers labeled $N_i$ for each $i \in [n]$ and $(n-1)$ layers labeled $J_i$ for each $i \in [n-1]$, and the layers are ordered from left to right in the following way: $N_1 J_1 N_2 J_2 ... J_{n-1} N_n$. Each $N_i$ has vertices $\{1,...,n\}$ and each $J_i$ has vertices $\{1,...,n\} \cup \{K_{i,j}^{(\ell)}\}_{j,\ell}$. The only edges between layers fall under two classes:

\begin{enumerate}
	\item ($N_i$ to $J_i$) For every $q \in [n]$, vertex $q$ in $N_i$ always connects to vertex $q$ in $J_i$. Vertex $i$ in $N_i$ connects to vertex $K_{i,j}^{(\ell)}$ in $J_i$ iff the value $K_{i,j}^{(\ell)}$ is $1$. 
	\item ($J_i$ to $N_{i+1}$) For every $q \in [n]$, vertex $q$ in $J_i$ always connects to vertex $q$ in $N_{i+1}$. Vertex $K_{i,j}^{(\ell)}$ in $J_i$ always connects to vertex $j+1$ in $N_{i+1}$.
\end{enumerate}

\Cref{ldagimg} shows an example of this construction where the elements of $B \in \{0,1\}^{3 \times 3}$ are defined implicitly in a $2 \times 2$ upper triangular matrix $\{(K_{i,j}^{(1)}, K_{i,j}^{(2)})\}_{1 \leq i \leq j \leq n-1}$:
\begin{equation*}
\begin{pmatrix}
    (1,0) & (1,1)\\
    0 & (0,1)\\
\end{pmatrix}
\end{equation*}
where, for example, $(K_{1,2}^{(1)}, K_{1,2}^{(2)}) = (1,1)$ is the top-right entry. 

\begin{figure}[ht]
	\centering
	\begin{tikzpicture}
    \node[shape=circle,draw=black] (10) at (0,9) {1};
    \node[shape=circle,draw=black] (20) at (0,8) {2};
    \node[shape=circle,draw=black] (30) at (0,7) {3};
    
    \node[shape=circle,draw=black] (11) at (4.5,9) {1};
    \node[shape=circle,draw=black] (21) at (4.5,8) {2};
    \node[shape=circle,draw=black] (31) at (4.5,7) {3};
    \node[shape=circle,draw=black] (k121) at (4.5,5.5) {$K_{1,1}^{(1)}$};
    \node[shape=circle,draw=black] (k122) at (4.5,4) {$K_{1,1}^{(2)}$};
    \node[shape=circle,draw=black] (k131) at (4.5,2.5) {$K_{1,2}^{(1)}$};
    \node[shape=circle,draw=black] (k132) at (4.5,1) {$K_{1,2}^{(2)}$};
    
    \node[shape=circle,draw=black] (12) at (8,9) {1};
    \node[shape=circle,draw=black] (22) at (8,8) {2};
    \node[shape=circle,draw=black] (32) at (8,7) {3};
    
    \node[shape=circle,draw=black] (13) at (11,9) {1};
    \node[shape=circle,draw=black] (23) at (11,8) {2};
    \node[shape=circle,draw=black] (33) at (11,7) {3};
    \node[shape=circle,draw=black] (k231) at (11,5.5) {$K_{2,2}^{(1)}$};
    \node[shape=circle,draw=black] (k232) at (11,4) {$K_{2,2}^{(2)}$};
    
    \node[shape=circle,draw=black] (14) at (14,9) {1};
    \node[shape=circle,draw=black] (24) at (14,8) {2};
    \node[shape=circle,draw=black] (34) at (14,7) {3};

    \path [->](10) edge node[left] {} (11);
    \path [->](11) edge node[left] {} (12);
    \path [->](12) edge node[left] {} (13);
    \path [->](13) edge node[left] {} (14);
    
    \path [->](20) edge node[left] {} (21);
    \path [->](21) edge node[left] {} (22);
    \path [->](22) edge node[left] {} (23);
    \path [->](23) edge node[left] {} (24);
    
    \path [->](30) edge node[left] {} (31);
    \path [->](31) edge node[left] {} (32);
    \path [->](32) edge node[left] {} (33);
    \path [->](33) edge node[left] {} (34);
    
    \path [->](10) edge node[left] {} (k121);
    \path [->](10) edge node[left] {} (k131);
    \path [->](10) edge node[left] {} (k132);
    
    \path [->](k121) edge node[left] {} (22);
    \path [->](k122) edge node[left] {} (22);
    \path [->](k131) edge node[left] {} (32);
    \path [->](k132) edge node[left] {} (32);
    
    \path [->](22) edge node[left] {} (k232);
    
    \path [->](k231) edge node[left] {} (34);
    \path [->](k232) edge node[left] {} (34);
\end{tikzpicture}
	\caption{Example:  creating a layered DAG from a DAG.}
	\label{ldagimg}
\end{figure}
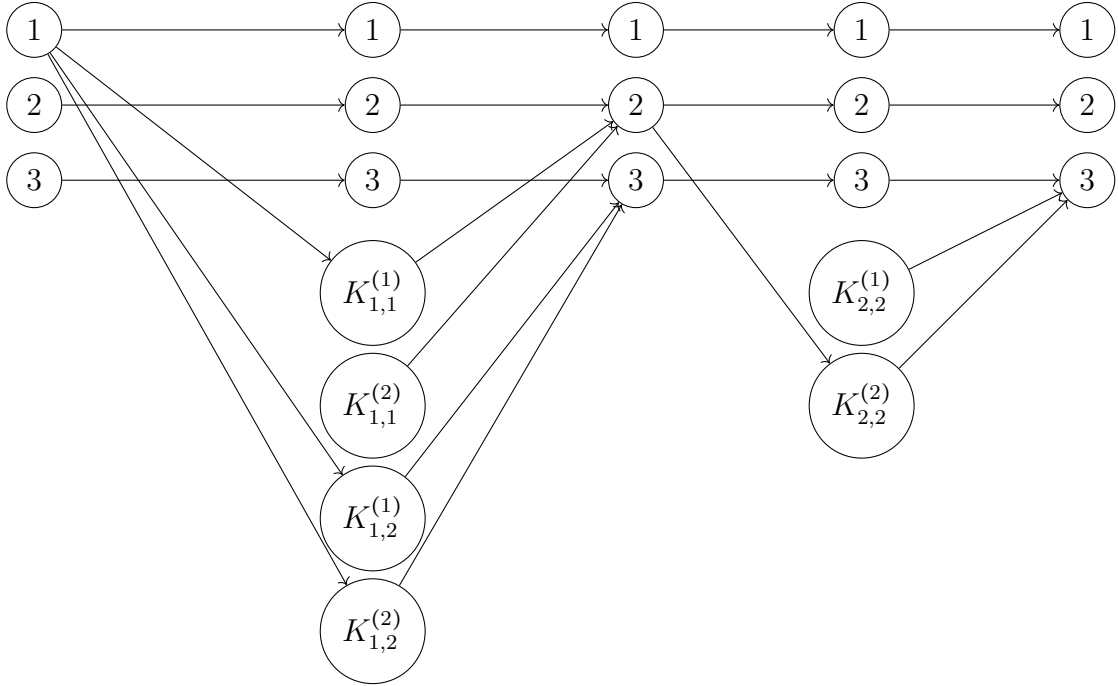

The following proposition shows that our construction of $C$ preserves the parity of paths from source to target:

\begin{proposition}
\label{prop:C}
Given $A \in M$, the parity of the number of paths in $C$ from vertex $1$ in $N_1$ to vertex $n$ in $N_n$ is equal to the parity of the number of paths from $1$ to $n$ in $A$.
\end{proposition}
\begin{proof}
We prove inductively that the parity of paths from vertex $i$ in $N_i$ to vertex $n$ in $N_n$ is the same as $i$ to $n$ in $B$. Note that the edge $i$ to $i$ between $N_i$ and $J_i$ does not contribute to the parity of paths from $i \in N_i$ to $n \in N_n$. The parity of paths from $i \in N_i$ to $j \in N_{i+1}$ is equal to $K_{i,j}^{\oplus}$, and there is only one path from $j \in N_{i+1}$ to $j \in N_j$. Hence, the paths from $i \in N_i$ to $j \in N_j$ only change the parity if $K_{i,j}^{\oplus}=1$ and the parity of paths from $j$ to $n$ in $B$ is odd (by assumption). Therefore, the parity of paths from vertex $i$ in $N_i$ to vertex $n$ in $N_n$ is equal to those from $i$ to $n$ in $B$.
\end{proof}

This concludes the correctness of $\DAlg$.  We are now ready to show that it is also half-randomizing.
\begin{theorem}
\label{thm:D_half}
$\DAlg$ is half-randomizing.
\end{theorem}

\begin{proof}
Let $\Pi_{yes}$ be the adjacency matrices with an odd number of paths from $1$ to $n$, and let $\Pi_{no}$ be those adjacency matrices with an even number of paths.  Notice that $|\Pi_{yes}| = |\Pi_{no}|$ because adding/removing the edge from $1$ to $n$ forms a bijection between the two sets.  Therefore, the subroutine of $\DAlg$ mapping $A \in M$ to $B$ is half-randomizing because $B$ is a uniformly random instance of the same parity due to \Cref{prop:B}.  Finally, the matrix $C$ is generated by an injective function on $B$, so the entire algorithm $\DAlg$ is half-randomizing.
\end{proof}

We now lift the half-randomization of $\DAlg$ to the composition $\gamma \circ \EAlg \circ \DAlg$:

\begin{cor}
	\label{Fhalfrandomizing}
	$\gamma \circ \EAlg \circ \DAlg$ is half-randomizing.
\end{cor}
\begin{proof}
    By \Cref{fproperties} and injectivity of $\EAlg$ (\Cref{Fselfreduction}), for distinct LDAG adjacency matrices $C \neq C'$, $\gamma \circ \EAlg (C)$ and $\gamma \circ \EAlg(C')$ are uniform distributions over disjoint sets of the same cardinality. Combined with the fact that $\DAlg$ is half-randomizing (\Cref{thm:D_half}), we immediately get that $\gamma \circ \EAlg \circ \DAlg$ is also half-randomizing. 
\end{proof}

In conclusion, $\EAlg \circ \DAlg$ maps odd parity adjacency matrices to gates that multiply to the 3-cycle and even parity adjacency matrices to gates that multiply to the identity.  Running the $\gamma$ algorithm on this sequence of gates produces an input to the first round of the 2-Round Graph State Measurement Problem.  The composition of all three algorithms is half-randomizing. 

\subsubsection{Proof of average-case \texorpdfstring{$\parityL$}{ParityL}-hardness}

Notice that in the previous section we never stated that $\gamma \circ \EAlg$ is half-randomizing, and indeed it is not half-randomizing. As a result, any error in a rewind oracle for 2-Round Graph State Measurement might concentrate on a few instances $\gamma \circ \EAlg(C)$ for some layered DAG $C$.  Indeed, this was the entire reason $\DAlg$ was introduced.  On the other hand, we now claim (proof in \Cref{appendixparitylapplications}) that if this type of concentration of error does not occur, then the rewind oracle \emph{can} be used to solve any instance of $\LDAGParity$.

\begin{proposition}
	\label{restateerror}
	Let $C \in \{0,1\}^{n \times n}$ be a layered DAG adjacency matrix and $\mathcal{R}$ be the rewind oracle for 2-Round Graph State Measurement on a $\poly(n) \times \poly(n)$ grid. If 
	\begin{equation*}
		\Pr[\text{$\exists \lambda \in I_2$ s.t. $\mathcal{R}$ fails on input $(x,\lambda)$} \,\,|\,\, x \gets \gamma \circ \EAlg(C)] < \frac{1}{21}
	\end{equation*}
	then a $(\BPAC^0)^{\mathcal{R}}$ circuit can compute the parity of the number of paths from $1$ to $n$ in $C$.
\end{proposition}

We are now ready to prove the main theorem of this section, which we restate below:
\begin{reptheorem}{parityLmain}
Let $\mathcal{R}$ be the rewind oracle for 2-Round Graph State Measurement Problem on a $\poly(n) \times \poly(n)$ grid promised that input is from the image of $\gamma \circ \EAlg \circ \DAlg(M)$ in the first round and $I_2$ in the second round. If $\mathcal{R}$ fails a uniformly random input in the promise with probability $\epsilon < \frac{1}{421}$ then 
\begin{equation*}
\mathsf{\parity L} \subseteq (\mathsf{BPAC^0})^{\mathcal{R}}.
\end{equation*}
\end{reptheorem}
\begin{proof}
We will show that a $\BPAC^0$-circuit with access to rewind oracle $\mathcal R$ can solve $\DAGParity$, which is $\parityL$-hard by \Cref{dagparity}. Suppose $A \in \{0,1\}^{n \times n}$ is the adjacency matrix of a monotone graph, representing an instance of $\DAGParity$. By \Cref{prop:C}, we will use $\DAlg$ to reduce $A$ to a random $\LDAGParity$ instance $C \in \{ 0, 1 \}^{\poly(n) \times \poly(n)}$, and then use \Cref{restateerror} to solve the $\LDAGParity$ instance. If the rewind oracle is always correct, then this strategy clearly solves the original $\DAGParity$ problem, so it only remains to analyze the allowable error. 

By assumption, with probability $\epsilon$ the rewind oracle $\mathcal{R}$ will produce a faulty output on an input sampled from $\gamma \circ \EAlg \circ \DAlg(M) \times I_2$, over the internal randomness of $\mathcal{R}$ and choice of input. Consequently, $\mathcal{R}$ fails with probability at most $(2 \epsilon)$ when input is sampled from $\gamma \circ \EAlg \circ \DAlg(A) \times I_2$ because $\gamma \circ \EAlg \circ \DAlg$ is half-randomizing (\Cref{Fhalfrandomizing}) and half-randomizing algorithms evenly spread out errors (\Cref{2eps}). Furthermore, since $|I_2|=5$, we get
\begin{equation*}
	\operatorname*{\mathbb{E}}_{r_d} [ \Pr[\exists \lambda \in I_2 \text{ s.t. } \mathcal{R} \text{ fails on input } (x,\lambda) \,\,|\,\, x \gets \gamma \circ \EAlg \circ \DAlg_{r_d} (A)] ] \leq 10 \epsilon
\end{equation*}
by a union bound. Applying Markov's inequality, we get
\begin{equation}
\label{layerprob}
	\Pr_{r_d} [ \Pr[\exists \lambda \in I_2 \text{ s.t. } \mathcal{R} \text{ fails on input } (x,\lambda) \,\,|\,\, x \gets \gamma \circ \EAlg \circ \DAlg_{r_d} (A)] \geq \frac{1}{21} ] \leq 210 \epsilon.
\end{equation}
 We can assume that the $(\BPAC^0)^{\mathcal R}$ circuit from \Cref{restateerror} succeeds with probability $1 - \alpha$ for any constant $\alpha > 0$. Using this circuit, we get
\begin{equation*}
	\Pr_{r_d} [\text{circuit fails to decide parity of } \DAlg_{r_d} (A)]
	\leq 210 \epsilon + \alpha (1-210\epsilon) 
	< \frac{1}{2}
\end{equation*}
where the first inequality comes from \Cref{restateerror} and \Cref{layerprob} and the final inequality comes from choosing small enough $\alpha$ and $\epsilon < \frac{1}{421}$. To finish, we repeat the process by sampling $O(\log n)$ elements from $\DAlg(A)$ and applying the circuit from \Cref{restateerror} on each sampled instance. By taking a majority vote of outputs over all instances, we recover the parity of the paths from $1$ to $n$ in $A$ with high probability.
\end{proof}

\section{Acknowledgments}
We acknowledge the support of the Natural Sciences and Engineering Research Council of Canada (NSERC). DG is also supported in part by IBM Research.  We also thank David Gosset and Alex Lombardi for useful discussions.

\printbibliography

\begin{appendices}
\crefalias{section}{appendix} 
\addtocontents{toc}{\protect\setcounter{tocdepth}{2}}
\makeatletter
\addtocontents{toc}{%
  \begingroup
  \let\protect\l@section\protect\l@subsection
  \let\protect\l@subsection\protect\l@subsubsection
}
\makeatother

\section{Improved \texorpdfstring{$\NC^1$}{NC\textonesuperior}-hardness error bound}
\label{improvednc1}

Here we improve the error threshold for the 2-Round Graph State Measurement problem on an $2 \times O(n)$ grid (see Section~\ref{sec:nc1hardness}) becomes hard. Previously, the threshold was $\epsilon < \frac{2}{75}$ \cite{gs}, and we improve it to $\epsilon < \frac{1}{30}$. 

\begin{lemma}
	Let $\mathcal{R}$ be the rewind oracle for the interactive problem in Section~\ref{sec:nc1hardness}. Suppose for a uniformly random input (first \& second round), $\mathcal{R}$ is incorrect with probability $\epsilon < \frac{1}{30}$. Then
\begin{equation}
	\NC^1 \subseteq (\mathsf{BPAC^0})^{\mathcal{R}}
\end{equation}
\end{lemma}
\begin{proof}
    Let $C_1, \ldots, C_n$ be two-qubit Cliffords. By \Cref{selfreduce}, we can use $\mathcal{R}$ to output one of fifteen non-trivial Paulis:\footnote{We remind the reader that we consider Paulis only up to sign, and exclude $II$, which is why there are fifteen rather than $32$ or $64$.} either a uniformly random non-stabilizer or a uniformly random stabilizer (excluding $II$) for the state $\ket{\psi} := C_1 \cdots C_n \ket{++}$. 
    
    Since \Cref{selfreduce} makes $6$ uniformly random calls to $\mathcal{R}$, and each one fails with probability at most $\epsilon$, there is a failure with probability at most $6 \epsilon$. Only when there is a failure can the algorithm output a stabilizer Pauli, and there are three nontrivial stabilizers of $\ket{\psi}$, so the probability of returning any particular stabilizer $S \in \mathcal{P}_2 \backslash \{ II \}$ is 
	\begin{equation}
		\label{eq1}
		\text{Pr[output $S$]} < 6 \epsilon \cdot \frac{1}{3}.
	\end{equation}
    On the other hand, at least $1 - 6\epsilon$ fraction of the time the algorithm returns one of twelve nonstabilizer Paulis. Thus, any particular nonstabilizer $N \in \mathcal{P}_2$ is observed with probability 
	\begin{equation}
		\label{eq2}
		\text{Pr[output $N$]} \geq (1 - 6 \epsilon) \cdot \frac{1}{12}.
	\end{equation}
	
	As long as $\epsilon < \frac{1}{30}$, we have that \Cref{eq1} and \Cref{eq2} are bounded above by $\frac{1}{15}$ and below by $\frac{1}{15}$, respectively. In fact, the gap between the bounds is a constant, $\sigma := (1 - 30 \epsilon) / 12$, so with $O(1 / \sigma^2)$ samples we can empirically estimate the probability for each Pauli well enough to distinguish stabilizers from nonstabilizers (with constant probability). It follows $(\mathsf{BPAC}^{0})^{R}$ circuit can learn the stabilizer group of $\ket{\psi}$, and thus solve an $\NC^{1}$ hard problem.
\end{proof}
\section{Error analysis for \texorpdfstring{$\parityL$}{Parity L}-hardness}
\label{erroranalysisappendix}

The purpose of this section is to fill in some of the details regarding the interactive protocol on the $\poly(n) \times \poly(n)$ grid used in \cite{gs}.  Along the way, will give some more detailed accounting of the error of that protocol.  Eventually, this will culminate in a proof of \Cref{restateerror}, the main missing ingredient in \Cref{sec:parityL_hardness} for $\parityL$-hardness (\Cref{parityLmain}).  This appendix is intentionally brief---for a slower introduction to the material, as well as proofs of correctness, we refer the reader to \cite{gs}.

We begin by borrowing much of the notation from \cite{gs}.  First, let us define two groups and a set of 3-qubit Paulis with important roles in the interactive protocol. Let $H_m = \langle \CZ, \text{S} \rangle_m$, i.e., circuits composed of $\CZ$ and Phase gates. Also let $H_3^{\oplus}$ be the subgroup of $H_m$ that consists of even numbers of $\CZ$ and S gates on the first three qubits. We have the following relationship: $H_3^{\oplus} \leqslant H_m \trianglelefteq G_m$. 

We denote the sets of second round measurements (Pauli measurements on the first three qubits) by
\begin{align*}
        I_2 = \{&\{XXX, XYY, YXY, YYX\}, \{IYI, XII, XYY, IIY\}, \{IIY, YXY, YII, IXI\}, \\
        &\{XXX, XII, IIX, IXI\}, \{IYI, IIX, YII, YXY\} \},
\end{align*}

and define $\medstar$ be the union over these sets. Let us define the operation $A \bullet B$ as $ABA^{-1}$. It can be checked that $|H_3^{\oplus} \bullet \medstar| = 24$ where $\bullet$ is performed pairwise between the elements of the two sets. 

Let us now review a protocol from \cite{gs} that uses the rewind oracle for the interactive task to glean some information about the second round state.  To be clear, this protocol is called several times in order to solve a single $\parityL$-hard instance.  There are two inputs to the problem: first, $g_1,...,g_n \in \CNOT_m$, which are promised to multiply to the 3-cycle or identity (identifying which is the $\parityL$-hard problem); second, an element $f \in H_3^{\oplus}$ which will be altered by the product $g_1 \cdots g_n$ via conjugation.  That is, the protocol can be used successfully to identify the product $g_1 \cdots g_n$ if it can track how the $f$ is changed.  For this reason, we label the algorithm $R_f$ so that it is parameterized by this important element $f$.  The output of $R_f$ is a single 3-qubit Pauli which is obtained by playing the magic pentagram game many times using the rewind oracle. We write out the full details of this protocol in pseudocode (\Cref{alg:algo1}).

\begin{algorithm}
\KwIn{$f \in H_3^{\oplus}$, $g_1,...,g_n \in \CNOT_m$ promised that $\pi = g_1...g_n \in \{C_3,I\}$}
\KwOut{A 3-qubit Pauli in $\medstar$}
	Sample $f' \leftarrow H_3^{\oplus}$\;
	Sample $h_1, ..., h_{2n-1} \leftarrow H_m$\;
	\tcc{The following is Kilian randomization}
	$(g_1',...,g_{2n}') \leftarrow (f'g_1h_1 , h_1^{-1}g_2h_2, ... , h_{n-1}^{-1}g_nh_n , h_n^{-1}fg_nh_{n+1} , ... , h_{2n-2}^{-1}g_2h_{2n-1} , h_{2n-1}^{-1}g_1)$\;
	Input $(g_1',...,g_{2n}')$ to $\mathcal{R}$ in first round\;
	\For{Pauli line $(P_1,P_2,P_3,P_4)$ in $I_2$}{\label{forins}
		Input $(P_1,P_2,P_3,P_4)$ to $\mathcal{R}$ in the second round\;
		Record measurement outcome of $P_i$ for $i \in [4]$ from $\mathcal{R}$\;
		Rewind $\mathcal{R}$ to beginning of second round\;
	} 
	\tcc{Each 3-qubit Pauli $P_i$ is measured exactly twice in the loop}
	$P \leftarrow$ any 3-qubit Pauli for which $\mathcal{R}$ returns inconsistent results\;
	\KwRet{$f'^{-1} \bullet P$}
 \caption{Randomized algorithm, $R_f$, carried out by $(\BPAC^0)^{\mathcal{R}}$ circuit} 
 \label{alg:algo1}
\end{algorithm} 

From \cite{gs}, we know that the first round input\footnote{Although each $g_i'$ appearing in the algorithm is an element of the form $H_m \CNOT_m H_m$, we can w.l.o.g. consider it in the form $\CNOT_m H_m$ due to the fact that for any $g \in \CNOT_m$, $H_m g H_m = g H_m$, and this transformation can be computed in $\NC^0$. See \cite{gs} for details.} $(g_1', ..., g_{2n}')$ in \Cref{alg:algo1} is uniformly random with the constraints: (1) $g_i'H_m = g_iH_m$ and $g_{n+i}'H_m = g_{n-i+1}H_m$ for $i \in [n]$ and (2) $g_1'...g_{2n}' \in H_3^{\oplus}$. Let us denote by $I_1(g_1,...,g_n)$ the support of this distribution. 

If the oracle makes no errors, then the output of $R_f$ will be a sample from the distribution of 3-qubit Paulis $(\pi f \pi^{-1}) \bullet \mathcal{D_R}$ where $\mathcal{D_R}$ is a distribution over the 20 nonstabilizers of $\ket{+^3}$ in $S := H_3^{\oplus} \bullet \medstar$ . Let us call this errorless distribution $R_{f,0}$. Otherwise, suppose that the oracle fails on $(g_1', ..., g_{2n}')$ with probability $\epsilon$ for $(g_1', ..., g_{2n}')$ sampled from $I_1(g_1,...,g_n)$, where we say the oracle fails for $(g_1',...,g_{2n}') \in I_1(g_1,...,g_n)$ if it fails for any input in $(g_1',...,g_{2n}') \times I_2$. Then with probability $\epsilon$, $R_f$ samples from any fixed distribution over $S$, and with probability $1-\epsilon$, $R_f$ samples from $R_{f,0}$. Let us call this faulty distribution $R_{f,\epsilon}$. Notice that in both cases, $\mathcal{D_R}$ is fixed regardless of the value of $f$. This leads us to choose $f$ in a useful way for distinguishing the 3-cycle or identity.

\begin{lemma}[Theorem 33, \cite{gs}]
\label{x-type}
For every Pauli $P \in \mathcal{D_R}$, an $\NC^0$ circuit can determine an $f \in H_3^{\oplus}$ such that $f \bullet P$ has no weight in $R_{f,0}$ if $\pi = C_3$.
\end{lemma}

A nice observation is that if we set $f$ to be the identity element, then \Cref{alg:algo1} \emph{is} sampling from $\mathcal{D_R}$.  After obtaining an element from this distribution, say $P$, we can use \Cref{x-type} to identify an $f$ such that the new output distribution will never output $f \bullet P$ if the unknown permutation $\pi$ is the 3-cycle. Therefore, by repeatedly sampling from $R_{f,0}$, we can learn $\pi$.  The following lemma shows that this protocol still works in the presence of error.

\begin{lemma}
	\label{avgcaseanalysis}
	Suppose that a $\BPAC^0$ circuit can sample from a distribution $\gamma_f(g_1,...,g_n)$ over first round inputs to $\mathcal{R}$ of the form in the first three lines of Algorithm \ref{alg:algo1} and
	\begin{equation*}
		\delta := \Pr[\exists \lambda \in I_2 \text{ s.t. } \mathcal{R} \text{ fails on input } (x, \lambda) \; | \; x \leftarrow \gamma_f(g_1,...,g_n)] < \frac{1}{21}
	\end{equation*}
	Then a $(\BPAC^0)^{\mathcal{R}}$ circuit can determine whether $\pi=g_1 \cdots g_n$ is the 3-cycle or identity.
\end{lemma}
\begin{proof}
	We divide the $(\BPAC^0)^{\mathcal{R}}$ circuit into two distinct components: first, the circuit learns information about the distribution $\mathcal{D_R}$; second, the circuit uses this information to deduce the value of $\pi$.
	
	In the first phase, the circuit samples $O(\log n)$ 3-qubit Paulis from $R_{I, \delta}$, where the oracle $\mathcal{R}$ fails with probability $\delta < \frac{1}{21}$ for first round input sampled from $\gamma_f(g_1,...,g_n)$ (by assumption). That is, the circuit samples from $R_{f,0} = \mathcal{D_R}$ with probability $1-\delta$. Since there are only 20 Paulis in the support of $\mathcal{D_R}$, there exists a 3-qubit Pauli $P \in S$ with maximal weight $\mathcal{D_R}(P) \geq \frac{1}{20}$. The Pauli $P$ also has the maximal weight in $R_{I, \delta}$ because $\frac{1}{20}(1 - \delta) \geq \delta$. By sampling $O(\log n)$ times from $R_{I, \delta}$ and choosing the most frequent Pauli, the circuit recovers the Pauli $P$ (or another element in the support of $\mathcal{D_R}$ with equal probability) with high probability. 
	
	In the second phase of the circuit, the circuit uses \Cref{x-type} to determine an $f \in H_3^{\oplus}$ such that $f \bullet P$ has no weight in $R_{f,0}$. Then, the circuit samples $O(\log n)$ Paulis from $R_{f, \delta}$. With probability $1-\delta$, a sample will be from the distribution $R_{f,0} = (\pi f \pi^{-1}) \bullet \mathcal{D_R}$. If $\pi = I$, then $R_{f,0} = f \bullet \mathcal{D_R}$, so $f \bullet P$ has equal weight in $R_{f,0}$ as $P$ has weight in $\mathcal{D_R}$. If $\pi = C_3$, then $f \bullet P$ has no weight in $R_{f,0}$ by \Cref{x-type}. Combining these facts, we get 
	\begin{equation*}
		R_{f, \delta \, : \, \pi=I}(f \bullet P) \geq (1-\delta) \mathcal{D_R}(P) \geq \frac{1}{21} > \delta \geq R_{f, \delta \, : \, \pi = C_3}(f \bullet P)
	\end{equation*}
	so $R_{f, \delta \, : \, \pi=I}(f \bullet P)$ and $R_{f, \delta \, : \, \pi = C_3}(f \bullet P)$ are at least a constant difference apart. Thus, the circuit can distinguish $\pi \in \{C_3, I\}$ using the $O(\log n)$ Pauli samples from $R_{f, \delta}$ with high probability by checking whether $f \bullet P$ appears in significantly more than a $1/21$ ratio of all samples. 
\end{proof}

\subsection{Properties of \texorpdfstring{$\gamma$}{gamma}}
\label{appendixparitylapplications}

Let us now connect back to the results of \Cref{sec:parityL_hardness}.  Let us define randomized algorithm $\gamma_f$ to be the first three lines of \Cref{alg:algo1}. Because the two constraints on the output distribution of $\gamma_f$ are actually independent of $f$, we will instead call this randomized algorithm $\gamma$ for convenience. In this section, we prove \Cref{fproperties} which concerns the output distribution of $\gamma$, and apply \Cref{avgcaseanalysis} to $\gamma$ to prove the error-tolerance described in \Cref{restateerror}. We begin with the proof of \Cref{fproperties}.

\begin{replemma}{fproperties}
	For any distinct sequences of CNOT gates $(a_1,...,a_n) \neq (b_1,...,b_n)$, $\gamma(a_1,...,a_n)$ and $\gamma(b_1,...,b_n)$ are uniform distributions over disjoint sets of the same cardinality.
\end{replemma}
\begin{proof}
Recall that for any $g_1, \ldots, g_n \in \CNOT_m$, the output $(g_1', \ldots, g_{2n}') \leftarrow \gamma(g_1,...,g_n)$ is uniformly random with the constraints: (1) $g_i'H_m = g_iH_m$ and $g_{n+i}'H_m = g_{n-i+1}H_m$ for $i \in [n]$ and (2) $g_1' \cdots g_{2n}' \in H_3^{\oplus}$. Constraint (1) above implies that the distributions $\gamma(a_1,...,a_n)$ and $\gamma(b_1,...,b_n)$ are disjoint. This is because for any distinct $a_i, b_i \in \CNOT_m$, we have that $a_iH_m \cap b_iH_m = \varnothing$. 
	
Furthermore, let's show that the supports of $\gamma(a_1,...,a_n)$ and $\gamma(b_1,...,b_n)$ have equal size. To show that $|\gamma(a_1,...,a_n)| = |\gamma(b_1,...,b_n)|$, we now show a bijection between the $a_1', \ldots, a_{2n}'$ and $b_1', \ldots, b_{2n}'$ output by $\gamma(a_1,...,a_n)$ and $\gamma(b_1,...,b_n)$.  By condition (1), we can write the outputs $a_1', \ldots, a_{2n}'$ such that $a_i' = a_i h_i$ and such that $a_1' \cdots a_{2n}' \in H_3^{\oplus}$.  By pushing all the elements of $H_m$ to the right, we get that
\begin{align*}
    a_1' \cdots a_{2n}' &= a_1 h_1 a_2 h_2 \cdots a_n h_n a_n h_{n+1} \cdots a_1 h_{2n} = h_1' \cdots h_{2n}',
\end{align*}
where say pushing $h \in H_m$ past elements $a_i \ldots a_1$ yields element $a_1 \ldots a_i h a_i \ldots a_1 \in H_m$.  Indeed, because $H_m$ is normal in $G_m$, we have that $h_i' \in H_m$.  Furthermore, by condition (2) we have $h_1' \cdots h_{2n}' \in H_3^{\oplus}$.  To obtain the image of the bijection we simply push back each $h_i'$ to the left in the sequence of CNOT gates $b_1, \ldots, b_n$:
$$
h_1' \cdots h_{2n}' = b_1 h_1'' b_2 h_2'' \cdots b_n h_n'' b_n h_{n+1}'' \cdots b_1 h_{2n}'' = b_1' \cdots b_{2n}'.
$$
Since all operations are invertible, this completes the bijection. 
\end{proof}

Finally, let us recall \Cref{restateerror}.

\begin{repproposition}{restateerror}
	Let $C \in \{0,1\}^{n \times n}$ be a layered DAG adjacency matrix and $\mathcal{R}$ be the rewind oracle for 2-Round Graph State Measurement on a $\poly(n) \times \poly(n)$ grid. If 
	\begin{equation*}
		\Pr[\exists \lambda \in I_2 \text{ s.t. } \mathcal{R} \text{ fails on input } (x, \lambda) \,\,|\,\, x \gets \gamma \circ \EAlg(C)] < \frac{1}{21}
	\end{equation*}
	then a $(\BPAC^0)^{\mathcal{R}}$ circuit can compute the parity of the number of paths from $1$ to $n$ in $C$.
\end{repproposition}
\begin{proof}
    The main difference between the statement above and \Cref{avgcaseanalysis} is the composition of $\gamma$ with $\EAlg$ (recall that we are referring to $\gamma_f$ as $\gamma$). We address this difference, which directly leads to the proof of the statement. 
    
    The algorithm $\EAlg$ is a valid and deterministic reduction from $\LDAGParity$ to $\CNOTMult$ (\Cref{Fselfreduction}). Thus determining whether the instance of $\CNOTMult$ output by $\EAlg(C)$ multiplies to the 3-cycle or identity also determines the parity of the number of paths from $1$ to $n$ in $C$. 
    
    Combining this fact with \Cref{avgcaseanalysis} directly implies the claim.
\end{proof}

\addtocontents{toc}{\endgroup}
\end{appendices}

\end{document}